\documentclass[lettersize,journal]{IEEEtran}
\pdfoutput=1
\usepackage{cite}
\usepackage{amsmath,amssymb,amsfonts}
\usepackage{graphicx}
\usepackage{subcaption}
\usepackage{makecell}
\usepackage{multirow,multicol}
\usepackage{balance}

\usepackage{textcomp}
\usepackage{xcolor}
\usepackage{amsthm,mathrsfs}
\usepackage[cmtip,all]{xy}

\usepackage[linesnumbered,ruled,vlined]{algorithm2e}
\usepackage{algpseudocode}  
\usepackage{amsmath}  

\def\BibTeX{{\rm B\kern-.05em{\sc i\kern-.025em b}\kern-.08em
    T\kern-.1667em\lower.7ex\hbox{E}\kern-.125emX}}

\newtheorem{mydef}{Definition}

\newtheorem{myeg}{Example}

\newtheorem{mythm}{Theorem}

\begin{document}

\title{Fast Answering Pattern-Constrained Reachability Queries with Two-Dimensional Reachability Index}

\author{Huihui Yang, Pingpeng Yuan
\thanks{Huihui Yang and Pingpeng Yuan are with the School of Computer Science and Technology, Huazhong University of Science and Technology, Wuhan 430074, China (e-mail: hh$\_$yang@hust.edu.cn; ppyuan@hust.edu.cn).}
\thanks{Manuscript received April 19, 2021; revised August 16, 2021.}}

\markboth{Journal of \LaTeX\ Class Files,~Vol.~14, No.~8, August~2021}%
{Shell \MakeLowercase{\textit{et al.}}: A Sample Article Using IEEEtran.cls for IEEE Journals}

\IEEEpubid{0000--0000/00\$00.00~\copyright~2021 IEEE}

\maketitle

\author{\IEEEauthorblockN{Huihui Yang, Pingpeng Yuan, Hai Jin}

\IEEEauthorblockA{
\textit{}, \\ 
\textit{Huazhong University of Science and Technology}\\
\text{Wuhan, China }\\
\{hh\_yang, ppyuan, hjin\}@hust.edu.cn}
}

\begin{abstract}
Reachability queries ask whether there exists a path from the source vertex to the target vertex on a graph. Recently, several powerful reachability queries, such as Label-Constrained Reachability (LCR) queries and Regular Path Queries (RPQ), have been proposed for emerging complex edge-labeled digraphs. However, they cannot allow users to describe complex query requirements by composing query patterns. Here, we introduce composite patterns, a logical expression of patterns that can express complex constraints on the set of labels. Based on pattern, we propose pattern-constrained reachability queries (PCR queries). However, answering PCR queries is NP-hard. Thus, to improve the performance to answer PCR queries, we build a two-dimensional reachability (TDR for short) index which consists of a multi-way index (horizontal dimension) and a path index (vertical dimension). Because the number of combinations of both labels and vertices is exponential, it is very expensive to build full indices that contain all the reachability information. Thus, the reachable vertices of a vertex are decomposed into blocks, each of which is hashed into the horizontal dimension index and the vertical dimension index, respectively. The indices in the horizontal dimension and the vertical dimension serve as a global filter and a local filter, respectively, to prune the search space. 
Experimental results demonstrate that our index size and indexing time outperform the state-of-the-art label-constrained reachability indexing technique on 16 real datasets. TDR can efficiently answer pattern-constrained reachability queries, including label-constrained reachability queries.
\end{abstract}

\begin{IEEEkeywords}
reachability query, pattern-constrained reachability, index, hash 
\end{IEEEkeywords}

\section{Introduction}
\label{sec:intro}
In recent years, graphs as a nonlinear data structure that can represent multiple complex relationships between entities have become increasingly popular and have gained widespread use in practical applications. 
Reachability queries, one of the fundamental operations on graphs, have received extensive attention due to their numerous applications in a wide range of fields, including event analysis \cite{DRS}, trajectory discovery \cite{SOIT}, clustering methodologies, and influence maximization \cite{KIM2017217}.
%

Most of the existing methods for reachability queries are limited to unlabeled graphs and focus on determining the existence of a path between two vertices. This line of research has been extensively explored over the past decades, yielding a rich body of work on unlabeled graph reachability \cite{2006'Dual, 2009'3-hop, 2011'Path-tree, chen'2008'chain, 2005'HLSS, 2005'HOPI, 2010'GRAIL, 2012'GRAIL, 2007'GRIPP, 2013'ferrari, 2014'IP, 2018'IP, 2013'TFLabel, 1989'TCC, 2017'BFL}. However, in many graphs, such as social networks, transport networks, and biological networks, edges are often assigned labels to denote various types of relationship between vertices. 

To leverage this rich semantic information, reachability queries have been extended to incorporate constraints on the labels of the edges along a path. A prominent paradigm for this is the \textbf{R}egular \textbf{P}ath \textbf{Q}uery (RPQ) \cite{Fan'2011}. An RPQ between a source vertex $u$ and a target vertex $v$ requires that the label sequence of the path from $u$ to $v$ satisfy a given regular expression. Consider the transportation network in \figurename~\ref{fig:motivation}, where vertices \textbf{A-F} represent geographical locations and edges are labeled with transportation modes (e.g., "\textit{rail}", "\textit{plane}", "\textit{bus}", "\textit{ferry}", "\textit{car}"). An RPQ from \textbf{A} to \textbf{D} with the constraint "$car^+ferry^+$" requires that a path contains both "\textit{car}" and "\textit{ferry}" labels, and the labels "\textit{car}" must appear before the labels "\textit{ferry}" in sequence. The path \textbf{A} $\to$ \textbf{C} $\to$ \textbf{F} $\to$ \textbf{D} with the label sequence [car, ferry, ferry] satisfies this constraint.
A fundamental and widely-studied special case of the RPQ is the \textbf{L}abel-\textbf{C}onstrained \textbf{R}eachability (LCR) query \cite{2017'LI,2020'P2H,2023'PDU,2022'P2H}. Instead of a complex regular expression, an LCR query specifies a set of allowed labels, and a valid path must use only labels from this set. For instance, in the same graph, an LCR query from \textbf{A} to \textbf{D} with the allowed label set $\{$"\textit{car}", "\textit{rail}"$\}$ would be satisfied by the path \textbf{A} $\to$ \textbf{C} $\to$ \textbf{E} $\to$ \textbf{D}, since all edges on this path are labeled with "\textit{car}".

\begin{figure}[!ht]
    \centering
    \includegraphics[width=\linewidth]{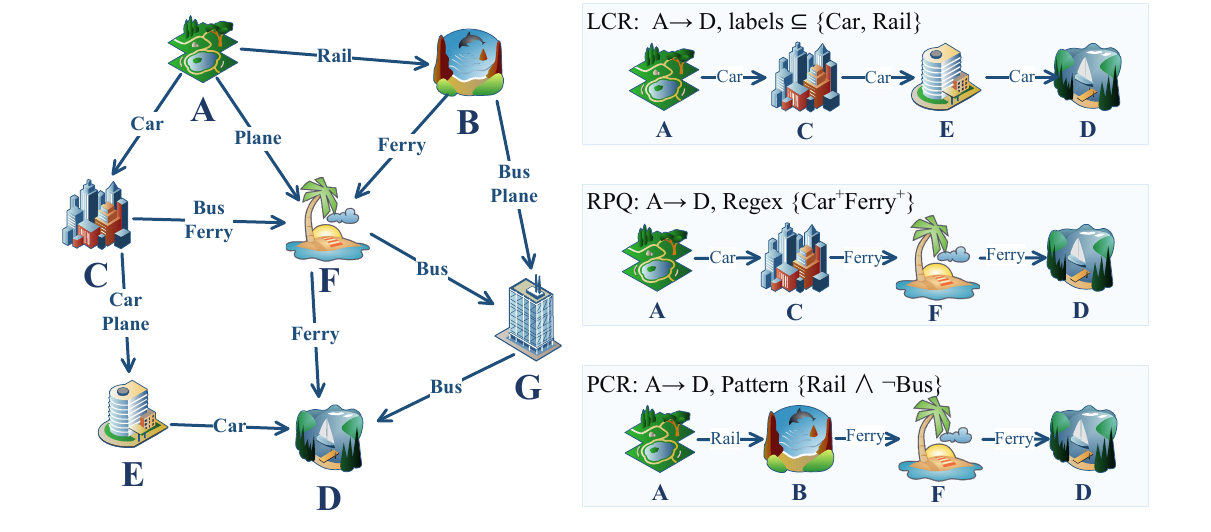}
    \caption{An illustrative example of three types of reachability queries on a graph of transportation network.}
    \label{fig:motivation}
    \vspace{-2mm}
\end{figure}

However, the expressive capabilities of both RPQ and LCR are often insufficient to model the complex, composite constraints required by many real-world applications. 
For instance, a traveler wishes to plan a journey from city \textbf{A} to city \textbf{D} in \figurename~\ref{fig:motivation} under the following constraints: the path must include a \textit{rail} segment for efficiency, but must exclude any \textit{bus} segment due to motion sickness. 
Formulating this requirement is challenging for existing paradigms. An RPQ excels at specifying sequential patterns but lacks native operators to express global negation. Conversely, an LCR query can only define a set of allowed labels but cannot enforce the mandatory presence of a specific label like \textit{rail} within that set. Consequently, neither RPQ nor LCR can succinctly encode this real-world constraint.
To address these expressiveness limitations, this paper introduces a novel \textbf{P}attern-\textbf{C}onstrained \textbf{R}eachability (PCR) query framework. In PCR, path constraints are formulated as propositional logic expressions over edge labels, seamlessly integrating logical conjunction ($\wedge$), disjunction ($\vee$), and negation ($\neg$). The aforementioned travel query can thus be precisely and intuitively expressed as \{ \textit{rail} $\wedge$ $\neg$\textit{bus} \}, directly mirroring the user's intent.

In PCR queries, the combinations of labels and vertices grows exponentially. Consequently, constructing indices that contain full reachability information tailored to the query pattern is costly. This challenge is exacerbated by the expansion of data resulting in larger graphs with multi-labeled edges, complicating the construction of efficient indices to resolve pattern-constrained reachability queries. Particularly in sparse graphs, the high cost of building full indices does not yield better query performance.
To optimize performance, it is crucial to strike a balance between the cost to build an index and query-answering overhead. Consequently, we introduce a lightweight \textit{\textbf{T}wo-\textbf{D}imensional \textbf{R}eachability (TDR)} index and then design an efficient algorithm  to quickly answer PCR queries.


The contributions of this paper are summarized as follows:
\begin{itemize}
\item \textbf{Pattern-Constrained Reachability Query}.  
We introduce the \textbf{P}attern-\textbf{C}onstrained \textbf{R}eachability (PCR) query, allowing users to define composite patterns using logical operators such as $\mathbb{AND}$, $\mathbb{OR}$. With composite patterns, users can specify more flexible and expressive patterns that differ from the rigid constraints of regular expressions on a solution path (RPQ), or are not merely limited to a set of labels (LCR queries). 
We prove that answering PCR queries is an NP-hard problem. As far as we know, we are the first to propose a pattern-constrained reachability query.
\item \textbf{Two-Dimensional Reachability Index}. We propose a two-dimensional reachability index which is built for each vertex to track all vertices it can reach and labels on the paths from the vertex. Given that a vertex typically has a large number of reachable vertices, its reachable vertices are decomposed into multiple independent blocks. This way allows for pruning the entire block out when the index of the block shows that it does not contain solutions, thus making the search space more manageable. Each block is then assigned to the horizontal and vertical dimensions of the TDR index, respectively. The horizontal dimension serves as a global filter, while the vertical dimension index prunes the search space according to local information. 
\item \textbf{Efficiency}. We conducted extensive experiments to compare our method with existing methods on a range of real and synthetic datasets. The results indicate that our method substantially decreases the time to answer PCR queries and can also effectively address the LCR queries.
\end{itemize}

\section{Related Work}\label{sec:relatedwork}
Early research efforts \cite{2006'Dual, 2009'3-hop, 2011'Path-tree, chen'2008'chain, 2005'HLSS, 2005'HOPI, 2010'GRAIL, 2012'GRAIL, 2007'GRIPP, 2013'ferrari, 2014'IP, 2018'IP, 2013'TFLabel, 1989'TCC, 2017'BFL} focus mainly on unlabeled graph. Since many real-world graphs are labeled on edges, recent efforts explore two kinds of reachability queries with label constraints: RPQ and LCR (Detailed analysis in Appendix \ref{appendix:relatedwork}). 
 
\textbf{Regular Path Query.} 
A Regular Path Query (RPQ) specifies that the labels of any solution paths must satisfy a given regular expression, where operators include union, concatenation, and Kleene closure \cite{2021'RQuBE}. 
RL \cite{2012'RPQonLG} answers RPQs by decomposing an RPQ into several smaller RPQs using rare labels. 
ARRIVAL \cite{2019'ARRIVAL} samples a number of paths through bi-directional random walk. If there exists any sampling path between two vertices, the two vertices are reachable, otherwise unreachable. So, the answer may be false-negative. 
Similarly, the example-based regular path query (RQuBE) \cite{2021'RQuBE} also employs a sampling-based method to build a candidate vertex set and return top-k vertices based on their confidence values as the final result set. 
Unlike previous sampling algorithms, D. Arroyuelo et al. \cite{2022'Ring} proposed a new algorithm that combines bit-parallel simulation and the ring index to process automaton states synchronously. In addition to generating vertices pairs from the constraints of regular expressions, there is some other work on regular path queries. For example, A. Pacaci et al. \cite{2020'Streaming} determine whether a given constraint is satisfied between two concrete entities over streaming graphs. Recent RLC queries \cite{2023'RLC} require solution paths consisting of one or more concatenations of the given sequence. 

\textbf{Label-constrained Reachability Query (LCR).} LCR query restricts the labels on the solution paths to only those within a given label set. 
LI+ \cite{2017'LI} answers LCR queries by precomputing  pairs of vertex that can be reached via the landmarks. Then, LI+ canthe stored information. Y. Chen et al. \cite{Chen'2021'Recurve} propose an algorithm that recursively decomposes the input graph while transforming the query into a series of subqueries to answer LCR queries. 
P2H+ \cite{2020'P2H,2022'P2H} stores all vertices that each vertex $u$ can reach or can reach $u$, along with the minimal set of labels on the path between $u$ and each reachable vertex. 
Since P2H+ requires huge storage, Y. Cai et al. \cite{2023'PDU} improve P2H+ by reducing the index size for vertices with one degree and eliminating unreachable queries without label constraints.
Since existing algorithms for LCR build complete indices, they offer the best performance on queries. However, their indexing costs (e.g. index size and indexing time) are relatively high, making them impractical for construction on large graphs. 


In real-world scenarios, users may expect the logical combinations of label sets (or patterns) rather than a label set (LCR) or regular expressions (RPQ). Therefore, we propose composite patterns using logical operators that relaxes these label constraints for more flexibility. 
In contrast to previous approaches, we construct a particle index answering whether a given vertex pair is reachable under a specified pattern.

\section{Preliminary}\label{sec:preliminary}

\subsection{Concepts and Definitions}
A real-world directed graph may be a multi-graph, where multiple relationships, each denoted by a unique label, can exist between any two entities. To facilitate a clear and unambiguous representation, we treat each distinct label as a separate edge. For example, the edge with multiple labels in \figurename~\ref{fig:example-graph}, $v_0 \xrightarrow{a,b} v_2$, is interpreted as two distinct edges between $v_0$ and $v_2$: one for the label '$a$' and another for the label '$b$'. 

\begin{figure}[t]
    \centering
    \vspace{-1mm}
    \begin{subfigure}{0.5\linewidth}
        \centering
    	\includegraphics[width=\linewidth]{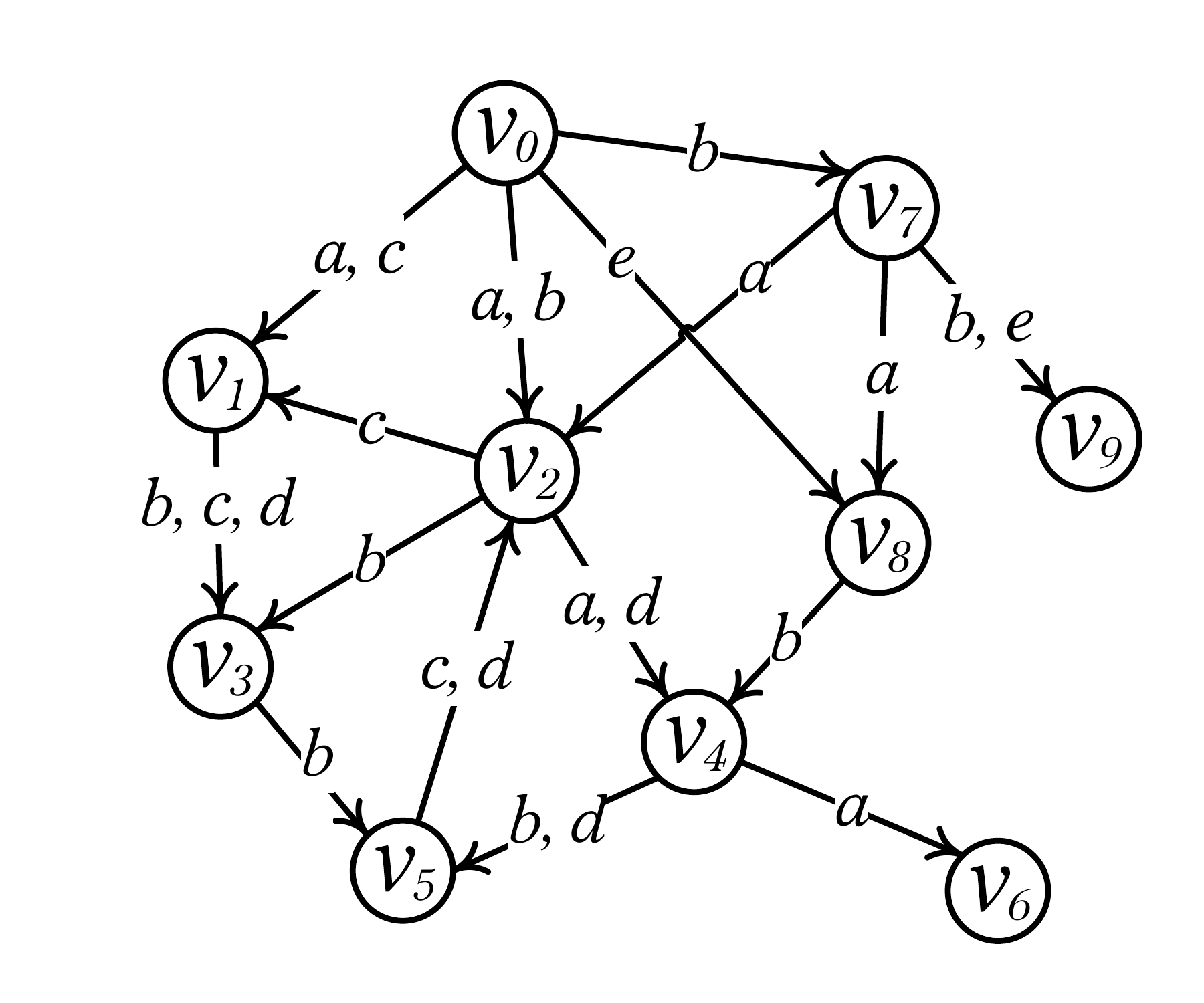}
        \vspace{-4mm}
        \caption{an edge-labeled digraph}
        \label{fig:example-graph}
    \end{subfigure}
    \begin{subfigure}{0.49\linewidth}
        \centering
    	\includegraphics[width=\linewidth]{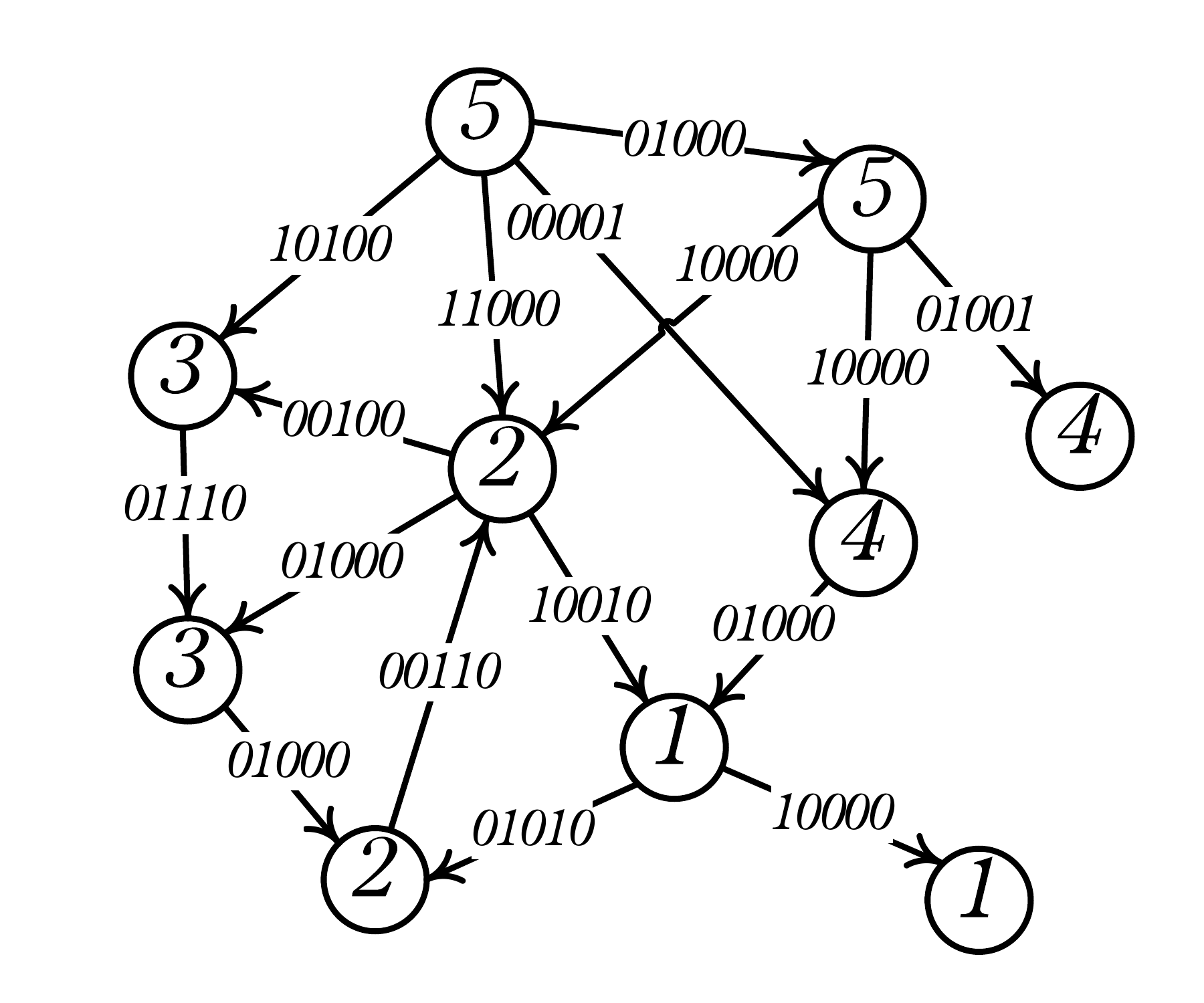}
        \vspace{-4mm}
        \caption{the hashed digraph}
        \label{fig:example-index}
    \end{subfigure}
    \vspace{-4mm}
    \caption{\small{An edge-labeled digraph with 10 vertices and 5 labels, and the digraph with vertices and labels hashed.}}
    \label{fig:example}
\end{figure}

\begin{mydef}[Edge-Labeled Digraph] An edge-labeled digraph is denoted as $G = (V, E, \zeta)$ where $V$, $E$ and $\zeta$ are the set of vertices, edges and labels, respectively. For each edge $e = \langle u, v, l \rangle \in E$, $u \in V$ is the source vertex, $v \in V$ is the target vertex, and $l \in \zeta$ is the label associated with the edge. 
\end{mydef}

Here, the number of vertices, edges, and labels in the graph $G$ are denoted as ${|V|}$, ${|E|}$, and ${|\zeta|}$, respectively. 
Let ${Suc(u) = \{ v | \langle u, v, l \rangle \in E \}}$,  ${Pre(u) = \{ v | \langle v, u, l \rangle \in E \}}$. If $Pre(u)$ is empty, then $u$ is the \textit{root vertex} (e.g., vertex $v_0$). Similarly, $u$ is the \textit{leaf vertex} if $Suc(u)$ is empty (e.g., vertex $v_6$). 

\begin{mydef}[Reachability]\label{def:reach} 
A path consists of a sequence of vertices and edges where the edges connect with each other, that is, $p: u_0$, $e_0$, $u_1$, $\dots$, $u_i$, $e_i$, $u_{i+1}$, $\dots$, $u_{m-1}$, $e_{m-1}$, $u_{m}$ where $e_i= \langle u_i, u_{i+1}, l_i \rangle \in E$ ($0\leq i<m$). 
We use $u_0 \stackrel{p}{\leadsto} u_{m}$ to denote that $u_0$ can reach \textbf{topologically} $u_{m}$ by the path $p$. 
Let $L(p)$ = $l_0l_1\dots l_i$ $l_{i+1} \dots l_{m-1}$ record the sequence of labels in $p$. 
$u_0$ can reach $u_{m}$ with label sequences $L(p)$, denoted as $u_0 \stackrel{L(p)}{\leadsto} u_{m}$. $u_0$ may reach $u_{m}$ by several paths. Let $\mathscr{L}(u_0 \leadsto u_{m})$ be the set of label sequences on the paths from $u_0$ to $u_{m}$, namely, $\mathscr{L}(u_0 \leadsto u_{m}) =\{L(p)|u_0 \stackrel{p}{\leadsto} u_{m}\}$. Similarly, we denote it by $ \xymatrix @C=5em {u_0\ar@{~>}[r]^{\mathscr{L}(u_0 \leadsto u_{m})} & u_{m}}$.
\end{mydef}

All paths from a vertex $u$ to the leaf vertices consist of a traversal tree rooted in $u$, denoted $T(u)$. Let \textbf{$V_{out}(u)$} denote the set of vertices that $u$ can reach, while \textbf{$V_{in}(u)$} is the set of vertices that can reach $u$. 
In some cases, reachability queries may not care about the order of labels on paths. So, we define the function $S()$ to return the set of labels corresponding to $L(p)$, that is, $S(L(p))=\{l|l\in L(p)\}$. Similarly, \textbf{$L_{out}(u)$} denote the set of labels on the paths by which vertices $u$ can reach, and \textbf{$L_{in}(u)$} is the set of labels on the paths to the target vertex $u$. That is, $L_{out}(u)=\{l|l \in S(\mathscr{L}(u \leadsto u_{i})), u_i \in V_{out}(u)\}$ and $L_{in}(u)=\{l|l \in S(\mathscr{L}(u_i \leadsto u)), u_i \in V_{in}(u)\}$.
%
\subsection{Problem Definition}
Here, we introduce composite patterns to help users define complex query constraints. 

\begin{mydef}[Pattern]
    A pattern $\mathcal{P}$ is a well-formed formula defined inductively over $\zeta$ as follows:
    \begin{enumerate}
        \item Atomic Pattern: If $l \in \zeta$, then both $l$ and $\neg l$ are well-formed atomic patterns. The former requires the presence of the label $l$, while the latter requires its absence. 
        \item Combination: If $\mathcal{P}_1$ and $\mathcal{P}_2$ are well-formed patterns, then $(\mathcal{P}_1)$ (Parenthesization), $\mathcal{P}_1 \wedge \mathcal{P}_2$ (Conjunction) and $\mathcal{P}_1 \vee \mathcal{P}_2$ (Disjunction) are also well-formed patterns. $\mathcal{P}_1 \wedge \mathcal{P}_2$ requires that both $\mathcal{P}_1$ and $\mathcal{P}_2$ be satisfied, while $\mathcal{P}_1 \vee \mathcal{P}_2$ requires that at least one of them be satisfied.
        \item Closure: All well-formed patterns are generated by finitely applying the rules (1) and (2).
    \end{enumerate}
\end{mydef}

If pattern $\mathcal{P}$ has only one label $l$, the constraint is met when $l$ appears.  Essentially, the presence of $l$ corresponds to \textbf{true} value of the logical expression. Consequently, for $\mathcal{P}$ with multiple labels interconnected by logical operators $\mathbb{AND}$, $\mathbb{OR}$, and $\mathbb{NOT}$, the constraint is satisfied by identifying a set of labels that makes the logical expression true. For example, given $\mathcal{P}=(l_1~\mathbb{AND}~l_2)~\mathbb{OR}~(\mathbb{NOT}~l_3)$, the pattern constraint is satisfied if either (1) both $l_1$ and $l_2$ are present or (2) $l_3$ is absent. Based on patterns, we define Pattern-Constrained Reachability Queries that allow users to specify complex constraints on queries.

\begin{mydef}[Pattern-Constrained Reachability Queries]
Given two vertices $u$, $v$, and a pattern $\mathcal{P}$, a \textbf{P}attern-\textbf{C}onstrained \textbf{R}eachability (PCR) query is to determine whether there exists a path from vertex $u$ to vertex $v$ such that the labels on the edges of the path satisfies the given pattern constraint $\mathcal{P}$, denoted as $ \xymatrix @C=1.5em {u\ar@{~>}[r]^?_{\mathcal{P}} & v}$. 



\end{mydef}

Unlike Regular Path Query (RPQ) \cite{Vertigo,RPQExample}, which defines the sequences of edge labels along solution paths, a PCR query identifies the set of labels present on solution paths. 
Let $match (\mathcal{P},\mathscr{L}(u \leadsto v))$ evaluate whether there exists a path between vertices $u$ and $v$ that conforms to the pattern constraints $\mathcal{P}$. If such a path $p$ exists, that is, if $S(L(p))$ satisfies the pattern $\mathcal{P}$, then $match (\mathcal{P},\mathscr{L}(u \leadsto v))$ = \textbf{true}. Otherwise, $match (\mathcal{P},\mathscr{L}(u \leadsto u))$ is \textbf{false}. 
Therefore, reachability queries $\xymatrix @C=1.5em {u\ar@{~>}[r]^?_{\mathcal{P}} & v}$ returns true if and only if both of the following conditions are met: $(a)$ $u \leadsto v$, i.e. topological reachability; $(b)$ $match (\mathcal{P},\mathscr{L}(u \leadsto v))$ = \textbf{true}, i.e. label reachability. 
In other words, when answering a reachability query, we can return false directly if topological or label reachability is false. 

\begin{myeg}
In \figurename~\ref{fig:example-graph}, $ \xymatrix @C=3.5em {v_0\ar@{~>}[r]^?_{b~\mathbb{AND}~d} & v_5}$ returns true because there exists a path $p:$ $v_0 \xrightarrow{a} v_1 \xrightarrow{d} v_3 \xrightarrow{b} v_5$ with $S(L(p))=\{a,b,d\}$ such that $match(\mathcal{P},\mathscr{L}(v_0 \leadsto v_5))$ = \textbf{true}. However, the solution for the PCR query $ \xymatrix @C=6em {v_0\ar@{~>}[r]^?_{\mathbb{NOT}(a~\mathbb{AND}~b)} & v_4}$ is false because no path exists between $v_0$ and $v_4$ that simultaneously excludes the labels $a$ and $b$.
\end{myeg}

It is known that any logical expression can be transformed into either the disjunctive normal form or the conjunctive normal form through a sequence of equivalent transformations. Likewise, a complex pattern can be decomposed into multiple sub-patterns connected by $\mathbb{OR}$, with each sub-pattern comprising labels connected by $\mathbb{AND}$. If each sub-pattern is denoted by a label, the original complex pattern can be represented through labels connected by $\mathbb{OR}$. For brevity, composite patterns that consist of labels, such as $l_0$ $\mathbb{AND}/\mathbb{OR}$ $l_1$ $\dots$ $\mathbb{AND}/\mathbb{OR}$ $l_m$ are denoted $\mathbb{AND}/\mathbb{OR}\{l_i\}_{0 \leq i \leq m}$. 
Therefore, for simplicity, here we only discuss the patterns where all logical operators are one of three logical operators, i.e. $\mathbb{AND}$, $\mathbb{OR}$, and $\mathbb{NOT}$. 
The logical operator of a pattern indicates whether the labels specified in its sub-patterns should be absent ($\mathbb{NOT}$) or present ($\mathbb{AND}$) in the solution path. 
For an $\mathbb{OR}$-pattern, if either of the subpatterns is met, the pattern is satisfied.


We prove that answering PCR queries is NP-hard (Theorem \ref{thr:complexity} in Appendix \ref{appendix:problem}). One straightforward method to answer PCR queries is to exhaustively traverse the graph \cite{Fan'2011}. However, this approach faces the challenge of searching through a vast number of permutations/combinations of labels and vertices if the graph is large. To address this, building the index which saves the permutations of labels on paths is a fast approach to answering the queries \cite{2020'P2H,2022'P2H}. However, it is time-consuming and requires a lot of storage because the permutations/combinations of labels are huge. 
It motivates us to design an index that consumes minimal storage while maintaining fast performance. A pattern specifies a set of labels which should match the sequence of label set on a solution. 


\section{Two-Dimensional Index}\label{sec:index}
Since PCR queries involve numerous combinations of labels, it is not feasible to maintain full reachable information for each vertex, as achieved with P2H+\cite{2020'P2H} and PDU\cite{2023'PDU}. Consequently, we opt to construct a partial index instead of a full index. 
There exist two challenges to build partial indices on large edge-labeled graphs. The first challenge involves efficiently pruning out non-viable branches. The subsequent challenge pertains to devising a compact index that can effectively scale to large graphs. 

\begin{figure}[htb]
    \centering   	
    \vspace{-2mm}
    \includegraphics[width=0.95\linewidth]{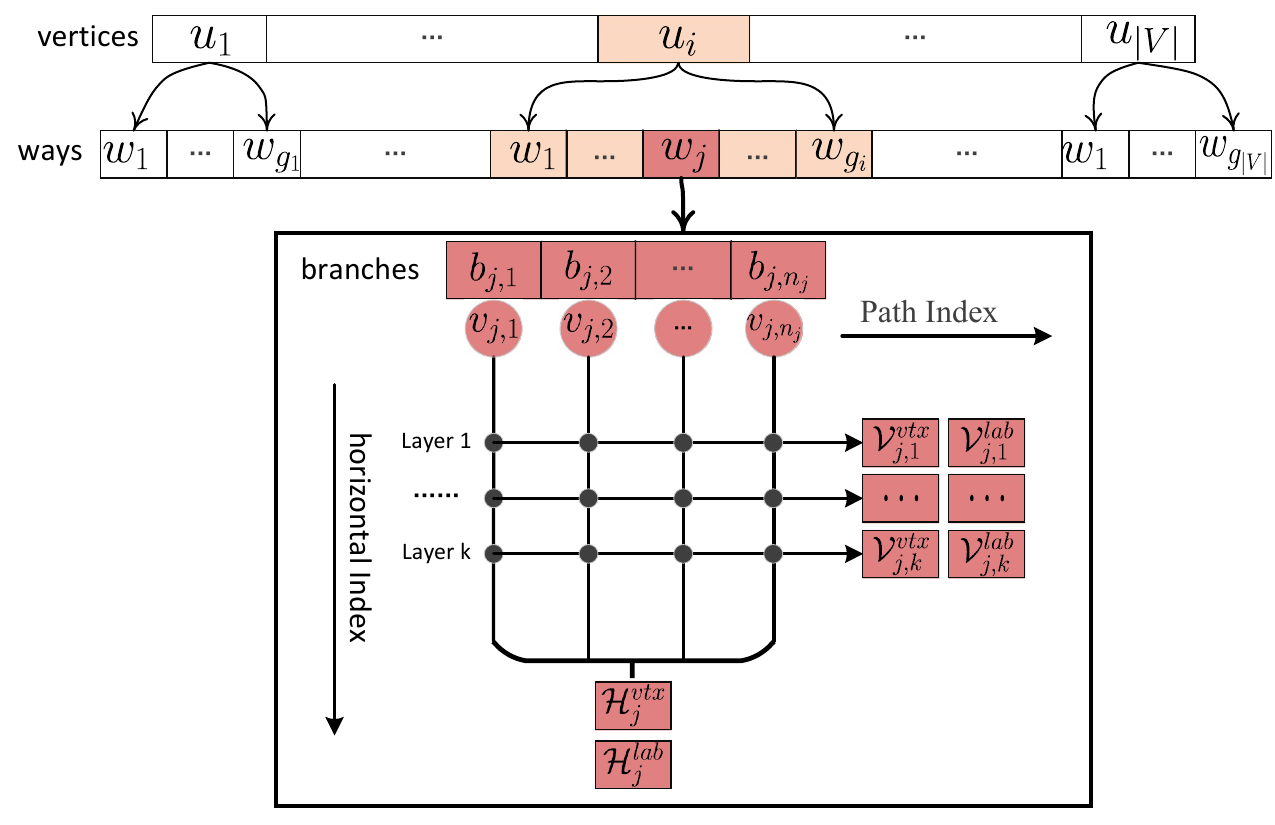} 	
    \vspace{-3mm}
    \caption{\small {The two-dimensional reachability index. The traversal tree starting from $u_i$ branches into $g_i$ distinct ways (denoted $w_1,\dots, w_{g_i}$), with each way (e.g. $w_j$) comprising one or more branches (e.g. $b_{j,1},\dots, b_{j,n_j}$) starting from neighbors of $u_i$ (e.g. $v_{j,1},\dots, v_{j,n_j}$). Each way is then projected onto both the horizontal and vertical dimensions, with the index in each dimension comprising two sub-indices for reachable vertex ($\mathcal{H}^{vtx}$,$\mathcal{V}^{vtx}$) and labels ($\mathcal{H}^{lab}$, $\mathcal{V}^{lab}$).}}
    \label{fig:frame}
    \vspace{-2mm}
\end{figure}

To solve the first challenge, we construct \textbf{T}wo-\textbf{D}imensional \textbf{R}eachability Index (TDR) for each vertex.  
The horizontal dimension of TDR indexes the reachable vertices and the label sets on the paths to them (Section \ref{sec:HD}), while the vertical dimension stores the vertex and label sequences on the fixed length paths (Section \ref{sec:VD}). The indices are further decomposed into multiple independent ways such that the ways can be pruned out when the index on the ways indicate they are not solutions. For the second challenge, the reachable label set, and the reachable vertex set are mapped into bit sets, respectively, which are stored in specially designed structures for the horizontal index and vertical index. The TDR index for each vertex $u$ can be summarized as follows: We decompose the traversal tree $T$ rooted in $u$ into multiple groups, each of which contains one or more branches of $T$. Then we compute the horizontal and vertical projection of each group, respectively. Concretely, for each group, we project horizontally to obtain the set of reachable vertices of $u$ and the set of labels on the path. Then the group is projected to vertical dimension, so the sequences of reachable vertices and labels are obtained, respectively.
\figurename~ \ref{fig:frame} illustrates this process.


\subsection{Multi-way Hashing for Horizontal Dimension}\label{sec:HD}
If ${u \leadsto v}$, then $u$ can reach all vertices that $v$ can reach, and all vertices that can reach $u$ can also reach $v$, that is, ${V_{out}(v) \subseteq V_{out}(u)}$ and ${V_{in}(u) \subseteq V_{in}(v)}$. Similarly, considering labels on paths, we also have ${L_{out}(v) \subseteq L_{out}(u)}$ and ${L_{in}(u) \subseteq L_{in}(v)}$. Based on these propositions, TDR index is constructed for each vertex to track all vertices and labels on the paths from the vertex. With the index, when answering PCR queries, we can prune search space. 

However, the reachable space is huge since the number of reachable vertices and labels of a vertex, particularly roots, may be large. When building reachability index for vertex $u$, we decompose the traversal tree that $u$ into multiple groups or ways such that each way can be handled independently. Thus, before building TDR, we must first determine the number of ways. It is not efficient for the traversal tree of each vertex to set the same number of ways. For example, the traversal trees of leaf vertices do not have branches, while root vertices generally have a large number of successors. Adopting the same number of ways for them may result in poor performance (more hash collisions or waste of storage). Thus, TDR dynamically generates the number of ways for each vertex based on its out-degree. To eliminate unnecessary indexing overhead, we do not build TDR indices for the vertices with zero out-degree because they cannot reach any vertices. 

\begin{algorithm}[!htp]
    \caption{Build Indexes}
    \label{alg:index}
    \KwIn{a root vertex $u$}
    \BlankLine
    $times\gets 0$, $s$ = $\emptyset$;     
    $s.push(u)$;
            
    \While{$s$ is not empty}
    {
        $w \gets s.top()$;
                
        \uIf {$v \in Suc(w)$ and $v$ is not visited }
        {
            $s.push(v)$;         
            $times \gets times+1$;
        
            record push time $\mathcal{I}_{push}(v) \gets times$;
        }                
        \Else
        {
            get the number of groups $g_w \gets hash(|Suc(w)|)$

            set the length of $\mathcal{H}^{vtx}(w)$ and $\mathcal{V}^{lab}(w)$ to $g_w$
            
            initialze $\mathcal{H}^{vtx}(w)$ to $hash(w)$
            
            \For {$v \in Suc(w)$}
            {    
            $i \gets v's~groupID$
            
            \textit{MultiWayHashing($w$, $v$, $i$)}
            
            \textit{PathHashing($w$,$v$, $i$)}
            }

            $s.pop()$
        
            $times \gets times+1$
        
            record pop time $\mathcal{I}_{pop}(w) \gets times$
        }
    } 

    \BlankLine
    \SetKwProg{PD}{Procedure}{:}{end}
\end{algorithm}

After obtaining a suitable way number $g$ for each vertex $u$, we can first build the horizontal dimension index, which is a set of bit masks and serves as a global filter that prune out explicitly unreachable vertices. we map $L_{out}$ and $V_{out}$ to $n$-way bit masks on average using Bloom filters, denoted by $\mathcal{H}^{lab}$ and $\mathcal{H}^{vtx}$ respectively.
The number of vertices is generally much larger than the number of labels in large graphs. Thus, it is infeasible to build reachability index for all reachable vertices and labels on the paths for each vertex using the same methods. The bit masks that store the vertices along each way will be longer than those storing the labels. However, since the length of the bit mask is much smaller than the size of $V_{in}$ or $V_{out}$, it is challenging to design collision-free hash functions. Therefor, we employ multiple simple hash functions. Conflicts may arise when the number of vertices is greater than the size of the set. To reduce the amount of conflict, instead of hashing by ID, we attempt to hash consecutive vertices along the path to the same hash value. This approach helps reduce the likelihood of hash collisions. 
Since the reachable vertices of a child are a subset of the reachable vertices of its parent, the building index follows a bottom-up approach by traversing  all children one by one. This process is repeated until all branches are visited. When visiting the children of a vertex, the indices of the children are combined into its reachable index using the logical operator $or$. Finally, the vertex itself is mapped into the bit mask. 

Algorithm~\ref{alg:index} constructs the horizontal dimension index using a bottom-up approach, initiating at the root vertex. It performs a deep traversal or backtracking by continually pushing (line 6) or popping from the stack (line 17). Each vertex is processed only after all its successors have been processed. When storing the information of all reachable vertices for a vertex $u$, the group ID of its successor is first determined according to a specified rule (line 12), and then all the bit masks in $\mathcal{H}^{vtx}$ of the successor are merged into the corresponding way of $u$ (line 13). In particular, the vertex itself will be hashed in each way (line 10). Similarly, the child label indices are also merged into its label index, including the edge labels between it and its children. 
The multi-way hash records the outgoing information of vertices. To acquire incoming relationship, we just reverse the traversal direction as delineated above. To avoid index redundancy, the number of ways for the reverse traversal is set to 1 and the labels are not saved. That is, we only need to hash $V_{in}$ to bit masks denoted as $\mathcal{N}_{in}$.



\subsection{Path Hashing for Vertical Dimension}\label{sec:VD}
Although the horizontal index can avoid unnecessary explorations by guiding the search, there exist some unnecessary traversals because the horizontal index is a global index, not a local index. For example, in \figurename~\ref{fig:example-graph}, for query $ \xymatrix @C=3.5em {v_0\ar@{~>}[r]^?_{\mathbb{NOT}\{b\}} & v_5}$ that $b$ must not appear in any solution paths, we need to search TDR until the index of vertex $v_3$ clearly indicates that there are no solutions through vertex $v_3$. To avoid these cases, we build the path index in the vertical dimension that contains only vertices within several hops. Different from the index of the horizontal dimension that globally filters the clearly unreachable vertices out, the index in the vertical dimension discard those branches according to several hops from the source vertex. 

The path index of each vertex stores the sequences of $k$ vertices and labels on fixed length paths, which share a common start vertex. Paths in a block may not be equal in length. Before building the path index, we must align multiple paths in a block so that each path is aligned with the other paths.  If the paths from a vertex (e.g. leaf vertices) have $m(m<k)$ elements, we will append the paths with $m-k$ virtual edges having null labels. This way facilitates the construction of the path index of a vertex from the path indices of its successors. Then, if the edges in the paths are equidistant from the start vertex, the edges are considered in a group, respectively. We merge the edge labels in the same group together. Labels on each level are represented using bit masks where each bit indicates whether the corresponding label appears in the level. 
The construction of the path index is also bottom-up. We use $k$-bit masks $\mathcal{V}^{lab}$ (\textit{PathHashing} of Algorithm~\ref{alg:index}) to store the first $k$-layers of all paths starting from $v$, $u$'s successors. After we build path indices for all successors of vertex $u$, then we can combine the path indices of its successors into its path index, and then add labels of its adjacent edges and neighbors into the top of the path index, respectively. Since the length of the path index is $k$, we slide the path index from top to bottom. Thus, the bottom elements are discarded. Thus, the length is still $k$ and the top elements are bit masks which correspond to labels between it and its successors. 
With the path index, we can prune some branches out without exploration if the branches do not match the pattern (Sec. \ref{sec:answerquery}). 

\begin{myeg}   
For \figurename~\ref{fig:example-graph}, there are 10 vertices $V=\{v_0,v_1,...,v_{9}\}$ and 5 labels $\zeta=\{a,b,c,d,e\}$. 
Both vertex and label hashing use 5-bit masks. $h()$ is the hash function for vertices. Labels are mapped to bits, with a bit value of 1 indicating the corresponding label presence. \figurename~\ref{fig:example-index} depicts the directed graph after vertex and label hashing. 
Let $v_7$ be a case and suppose $g_7$=2 and $k=1$. For the first way, which consists of branches that start from vertices $v_2$ and $v_8$, the layer-1 path hashing index is $\mathcal{V}^{vtx}_{1,1}(v_7)=hash(\{v_2,v_8\})=\{2,4\}$ and $\mathcal{V}^{lab}_{1,1}(v_7)=10000$. Similarly, the hash arrays in the horizontal dimension are $\mathcal{H}^{lab}_{1}(v_7)=11110$ and $\mathcal{H}^{vtx}_{1}(v_7)=\{1,2,3,4\}$. The second way contains only a vertex $v_9$, so the indices in both the horizontal and vertical dimension are consistent, as shown in Table~\ref{tab:example}.
\end{myeg}

\begin{table}
\centering
\renewcommand\arraystretch{1}
\caption{Index for \figurename~\ref{fig:example}}
\label{tab:example}
\resizebox{0.95\linewidth}{!}{
\begin{tabular}{l|c|c|c|c}
\hline
\textbf{VID}&\textbf{$\mathcal{H}^{vtx}$}&\textbf{$\mathcal{H}^{lab}$} & \textbf{$\mathcal{V}^{vtx}_{1},\mathcal{V}^{lab}_{1}$} & $[\mathcal{I}_{push},\mathcal{I}_{pop}]$\\
\hline 
\multirow{2}{*}{$v_0$} & $\{1,2,3\}$ & 11110 & $\{2,3\}$, 11100 & \multirow{2}{*}{[0,19]}\\
~ & $\{1,2,3,4,5\}$ & 11111 & $\{4,5\}$, 01001 & ~ \\
\hline
$v_4$ & $\{1,2\}$ & 11010 & $\{1,2\}$, 11010 & [5,8]\\
\hline
$v_6$ &   &   &   & [6,7]\\
\hline
\multirow{2}{*}{$v_7$} & $\{1,2,3,4\}$ & 11110 & $\{2,4\}$, 10000 & \multirow{2}{*}{[13,18]}\\
~ & $\{4\}$ & 01001 & $\{4\}$, 01001 & ~\\
\hline
$v_8$ & $\{1,2\}$ & 11010 & $\{1\}$, 01000 & [14,15]\\
\hline
\end{tabular}
}
\end{table}


\section{Answering Reachability Queries}\label{sec:answerquery}



When answering reachable queries, we exhaustively search all branches one by one and check if the labels on the path satisfy the pattern. We adopt a stack $s$ to remember the vertices that need to be explored (see Algorithm~\ref{alg:query}). 
We start from the source vertex $u$ (line 2) and push new matching vertices into the stack after exploring the next vertex by checking their indices. To determine if the pattern matches the labels, each bit mask generated from the given pattern is evaluated. If the label matches the pattern, add it to the matched label set and recursively check if this traversal leads to a valid path from source to target vertex. If the edges chosen in the above steps do not lead to a valid solution, it will perform a backtracking and remove these edges and labels from the candidate path and try other alternative edges. If none of the alternatives work, then it returns $unreachable$. The previously added labels and edges in recursion will be removed. If the initial call of recursion returns $unreachable$ then the final answer is also $unreachable$. 
They traverse all possible paths and then terminate the process when the answer is known. So, the search complexity grows exponentially with increasing graph size. 
This is inefficient because it is possible to perform many unnecessary traversals. To speed up the process, our approach adopts both block pruning, skipping, and early stopping to reduce search space. The block pruning reduces search space with multi-way index while the forward checking uses path index. 

\begin{algorithm}[htb]
    \caption{Answering Reachability Query}
    \label{alg:query}
    \KwIn{vertex $u$, $v$, and pattern $\mathcal{P}$}
    
    $s.push(u)$; 
    $l_{ptn} \gets hash($the labels in $\mathcal{P})$; 
    $l_{path} \gets \emptyset$
            
    \While{$s$ is not empty}
    {
        $m \gets s.top()$

        \uIf{$m==v$}
        {
            \Return $reachable$
        }

        \uIf{\textit{VertexReach(m,v)} or \textit{LabelReach(m)} is false}
        {
            $s.pop()$, go back to the last vertex 
            
            remove the label $\lambda(\langle m, v, l \rangle)$ from $l_{path}$
        }
            
        \For{each group $g_i$ of $m$}{
            \If{$\mathcal{N}_{out}(v) \subseteq \mathcal{H}^{vtx}(m)[i]$ and $\mathcal{H}^{lab}(m)[i]$ match $l_{ptn}$}
            {
                \If{$\omega \in g_i$ is not visited}
                {
                    $s.push(\omega)$

                    $l_{path} \gets l_{path}~|~hash(\lambda(\langle m, \omega, l \rangle))$
                }
            }
        }
    }

    \Return $unreachable$

    \BlankLine
    \SetKwProg{PD}{Procedure}{:}{end}
\end{algorithm}

\textbf{Group pruning.} As mentioned in the previous section,  the reachable vertices of a vertex are decomposed into one or more groups. Then each group is mapped to the horizontal index and vertical index, respectively. When answering reachability queries, the algorithm will first check the multi-way index of a vertex (lines 12 to 16). If the index does not match the pattern, it means that the group does not contain the target vertex or the labels on paths in the group does not satisfy the constraints specified in the pattern. Thus, the group will be pruned out. The vertices in the group will not be checked. By this way, we avoid unnecessary checks. Moreover, if the labels that are not allowed appear in the path index of the current vertex, the groups will be discarded without exploration because they do not match the pattern.

\textbf{Skipping label check after the pattern is satisfied.} At each vertex, the algorithm  evaluates whether combinations of labels stored in the path index of a vertex match the remaining query pattern (Procedure \textit{LabelReach} of Algorithm~\ref{alg:query}).  
If the remaining pattern is satisfied, the algorithm will skip the check of label reachability, but focus on answering topological reachability. So, it avoids checking labels of the blocks and thus can reduce the search space.

\textbf{Early stopping}. 
Since the path index can look several hops ahead, as the traversal nears the leaf vertices, if it shows that the branches do not belong to solution paths, we can immediately cease exploring that branch. Consequently, there is unnecessary to navigate each branch down to the leaf vertices. Moreover, we determine the topological reachability between vertices using $\mathcal{N}_{in}$ and the interval as detailed in the \textit{VertexReach} procedure of Algorithm~\ref{alg:query}. If the vertices are not topologically reachable, meaning no path connects them, the query will stop immediately and output "un-reachable". 

\begin{myeg} 
Considering TDR index in Table \ref{tab:example}, when answering query $\xymatrix @C=3em {v_7\ar@{~>}[r]^?_{\mathbb{NOT}~a} & v_4}$, the first way will be discarded since $\mathcal{V}^{lab}_1(v_7)=10000$, which indicates that there is no path without label $a$. Because of $hash(v_4)=1 \not \in \mathcal{H}^{vtx}_2(v_7)$, the second way will also be discarded. So the query returns "un-reachable". When addressing query $\xymatrix @C=3em {v_0\ar@{~>}[r]^?_{b~\mathbb{AND}~e} & v_6}$, since $\mathcal{H}^{lab}_1(v_0)=11110$, which means no paths in the first way contains label $e$, traversal proceeds from the second way. Along the path $v_0 \xrightarrow{e} v_8 \xrightarrow{b} v_4$, the labels $b$ and $e$ have been matched. Therefore, it is only necessary to verify whether $v_4$ is physically reachable to $v_6$. As interval of $v_6$=$[6,7]$ is contained within the interval of $v_4$=$[5,8]$, the outcome of this query is "reachable".

\end{myeg}

\section{Experimental Evaluation}\label{sec:eval}

\subsection{Experimental Settings}\label{exp:setting}

\textbf{Datasets:} 10 real datasets from SNAP \cite{SNAP} and KONECT \cite{KONECT} have $352K$ to $632M$ edges and $5 - 2774$ labels (Table \ref{tab:rg}). 
Furthermore, we produce synthetic graphs based on the '\textbf{P}referential \textbf{A}ttachment' model (PA-dataset), known for its skewed out-degree distribution \cite{1999'PA}, and the “\textbf{E}rd$\Ddot{o}$s-\textbf{R}$\Acute{e}$nyi” model (ER-dataset) , which approaches a uniform out-degree distribution \cite{1959'ER}. 
Each synthetic graph contains roughly 200K vertices. For each unlabeled graph, we produce labels that are uniformly assigned to its edges. 

\begin{table}
\centering
\caption{\small{Statistics of real digraphs}}
\label{tab:rg}
\resizebox{0.95\linewidth}{!}{
\begin{tabular}{l|rrcc}
\hline
\textbf{Dataset} & \textbf{$|V|$} & \textbf{$|E|$} & \textbf{$|\zeta|$} & \textbf{\thead{Synthetic\\Labels}} \\
\hline 
Youtube	&	15,089 	&	13,628,895 	&	5 	&		\\
StringFC	&	19,173 	&	6,513,176 	&	9 	&		\\
email	&	265,214 	&	418,956 	&	16 	&	$\surd$	\\
webStanford	&	281,904 	&	2,312,497 	&	32 	&	$\surd$	\\
NotreDame	&	325,729 	&	1,469,679 	&	16 	&	$\surd$	\\
citeseer	&	384,414 	&	1,744,590 	&	16 	&	$\surd$	\\
webBerkStan	&	685,231 	&	7,600,595 	&	32 	&	$\surd$	\\
wikitalk	&	1,140,149 	&	4,010,611 	&	2,321 	&		\\
socPokecL	&	1,632,804 	&	30,622,564 	&	32 	&	$\surd$	\\
twitter	&	41,652,231 	&	632,007,285 	&	32 	&	$\surd$	\\

\hline
\end{tabular}
}
\end{table}
\vspace{-1pt}

\textbf{Algorithm:} We are the first to investigate PCR queries on labeled graphs. For comparison, we implement a \textbf{DFS}-based approach to answer PCR queries because BFS is memory intensive. Current similar research efforts mainly focus on \textbf{L}abel-\textbf{C}onstrained \textbf{R}eachability (LCR) queries. 
PCR can express LCR queries using operators $\mathbb{NOT}$ and $\mathbb{AND}$. Here, we compare our approach with P2H+ \cite{2020'P2H,2022'P2H} and PDU (P2H+DOR+UQF) \cite{2023'PDU} on answering LCR queries.

\textbf{Query Generation:} We generate queries with two labels for the datasets \textit{Youtube} and \textit{StringFC}, while with four labels for all other datasets. For each dataset, $2k$ true-queries and $2k$ false-queries are generated based on the respective operators, and named as $\mathbb{AND}$-queries, $\mathbb{OR}$-queries, $\mathbb{NOT}$-queries and LCR-queries, respectively.

\textbf{Settings:} All experiments are run on a Linux server with 256GB of memory and a 2.0GHz Intel Xeon Gold 5117 CPU. All programs are implemented in c++ and compiled by g++ 7.5.0. 

\begin{table*}[htb]
\centering
\caption{\small{The execution time of $\mathbb{AND}$-, $\mathbb{OR}$-, and $\mathbb{NOT}$-queries on real digraphs for TDR and DFS (in second). Here, the number of labels in true-query set and false-query set is $|\zeta|/4$ or 4.}}
\label{tab:PCR}
\resizebox{0.95\linewidth}{!}{\begin{tabular}{l|rr|rr|rr|rr|rr|rr}
\hline
\multirow{3}{*}{\textbf{Dataset}} & \multicolumn{4}{c|}{\textbf{$\mathbb{AND}$}}&\multicolumn{4}{c|}{\textbf{$\mathbb{OR}$}}&\multicolumn{4}{c}{\textbf{$\mathbb{NOT}$}} \\
\cline{2-13}
~ & \multicolumn{2}{c|}{\textbf{true-query}}&\multicolumn{2}{c|}{\textbf{false-query}} & \multicolumn{2}{c|}{\textbf{true-query}}&\multicolumn{2}{c|}{\textbf{false-query}}& \multicolumn{2}{c|}{\textbf{true-query}}&\multicolumn{2}{c}{\textbf{false-query}} \\
~ & TDR & DFS & TDR & DFS & TDR & DFS & TDR & DFS & TDR & DFS & TDR & DFS \\
\hline 
Youtube	&	0.17 	&	56.15 	&	3.28ms	&	87.29 	&	0.02 	&	54.47 	&	1.91ms	&	85.78 	&	0.52 	&	46.87 	&	0.26 	&	44.97 	\\
StringFC	&	2.65 	&	4.16 	&	0.05 	&	36.72 	&	0.01 	&	3.77 	&	2.05ms	&	37.01 	&	0.50 	&	21.17 	&	0.17 	&	14.61 	\\
email	&	0.07 	&	67.02 	&	2.59ms	&	1.11 	&	3.32ms	&	52.00 	&	2.76ms	&	1.06 	&	0.84 	&	8.79 	&	0.01 	&	4.53 	\\
webStanford	&	0.03 	&	43.44 	&	0.06 	&	1.33 	&	3.74ms	&	42.27 	&	0.05 	&	1.36 	&	9.43 	&	19.08 	&	0.05 	&	0.26 	\\
NotreDame	&	0.02 	&	127.28 	&	0.02 	&	4.76 	&	3.32ms	&	118.56 	&	0.02 	&	4.51 	&	0.19 	&	2.05 	&	0.01 	&	0.54 	\\
citeseer	&	0.03 	&	33.59 	&	2.86ms	&	0.77 	&	0.02 	&	27.91 	&	3.22ms	&	0.65 	&	0.01 	&	1.15 	&	2.89ms	&	0.36 	\\
webBerkStan	&	0.03 	&	35.60 	&	0.08 	&	1.70 	&	0.01 	&	33.62 	&	0.07 	&	1.65 	&	3.49 	&	7.96 	&	0.05 	&	0.24 	\\
wikitalk	&	94,027.91 	&	197,423.10 	&	12,103.19 	&	39,525.68 	&	2,091.48 	&	24,999.64 	&	4,068.26 	&	17,746.65 	&	491.91 	&	9,918.00 	&	20.66 	&	7,301.04 	\\
socPokecL	&	1.62 	&	5,085.03 	&	2.98ms	&	2.64 	&	3.93ms	&	5,399.55 	&	3.27ms	&	3.84 	&	871.39 	&	1,376.85 	&	0.93 	&	13.23 	\\
twitter	&	61.54 	&	17,479.73 	&	3.03ms	&	30.07 	&	30.23 	&	20,101.39 	&	0.01 	&	14.88 	&	3,215.12 	&	6,255.80 	&	49.65 	&	325.75 	\\

\hline
\end{tabular}
}
\end{table*}

\subsection{Index}\label{exp:index}

\textbf{Index Time:} 
P2H+ and PDU timeout when processing large datasets, such as \textit{socPokecL} and \textit{twitter}. For datasets where three methods construct their indexes successfully, our approach is 2 to 5 orders of magnitude faster than PDU (e.g., \textit{citeseer}) and 3 to 6 orders of magnitude faster than P2H+ (e.g., \textit{StringHS}). When the graph is larger, our method requires considerably less time compared to P2H+ and PDU. The primary reason is that P2H+ requires multiple rounds of BFS to obtain the minimum label set between pairs of vertices during index construction. The construction process is very time-consuming when the graph is huge. The same applies to PDU. Instead, we construct partial indexes that require only a handful of depth-first searches to complete. Therefore, our method offers superior performance in processing large datasets. P2H+ requires more time to construct indexes for directed acyclic graphs, such as $String*$ because its pruning strategies do not work well in directed acyclic graphs. However,  our method remains unaffected by the directed acyclic nature of the graph. Consequently, the construction time of the index is relatively stable.

\begin{table}
\centering
\renewcommand\arraystretch{1}
\setlength\tabcolsep{2pt}
\caption{\small{Indexing time (IT) and indexing space (IS).  “-” indicates that the method times out or is out of memory on this dataset.}}
\label{tab:index}
\resizebox{0.95 \linewidth}{!}{
\begin{tabular}{l|rrr|rrr}
\hline
\multirow{2}{*}{\textbf{Dataset}}&\multicolumn{3}{c|}{\textbf{Indexing Time($s$)}}&\multicolumn{3}{c}{\textbf{Indexing Space($MB$)}} \\
\cline{2-7}
~ & P2H+ & PDU & TDR & P2H+ & PDU & TDR \\
\hline 
Youtube	&	1,489.32 	&	97.86 	&	0.34 	&	348.72 	&	14.34 	&	10.19 	\\
StringFC	&	9,826.35 	&	383.52 	&	0.18 	&	929.00 	&	152.14 	&	8.03 	\\
email	&	9.89 	&	0.82 	&	0.04 	&	101.70 	&	10.68 	&	4.13 	\\
webStanford	&	-	&	16.12 	&	0.28 	&	-	&	89.70 	&	6.89 	\\
NotreDame	&	1,890.73 	&	24.25 	&	0.12 	&	1,243.64 	&	115.56 	&	14.17 	\\
citeseer	&	-	&	8,205.92 	&	0.23 	&	-	&	6,120.11 	&	22.47 	\\
webBerkStan	&	-	&	66.03 	&	0.38 	&	-	&	342.78 	&	17.12 	\\
wikitalk	&	-	&	-	&	620.91 	&	-	&	-	&	14.69 	\\
socPokecL	&	-	&	-	&	4.05 	&	-	&	-	&	43.10 	\\
twitter	&	-	&	-	&	108.39 	&	-	&	-	&	861.68 	\\

\hline
\end{tabular}
}
\end{table}

\begin{figure*}[htbp]
    \centering
    \vspace{-3mm}
    \begin{subfigure}{0.6\linewidth}
        \centering
        \includegraphics[width=0.95\linewidth]{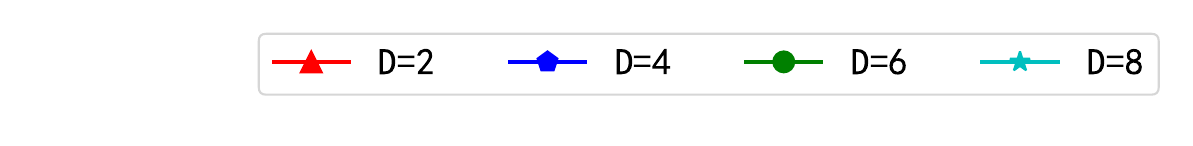}
    \end{subfigure}
    \vspace{-3mm}
        
    \begin{subfigure}{0.19\linewidth}
	\centering
    	\includegraphics[width=\linewidth]{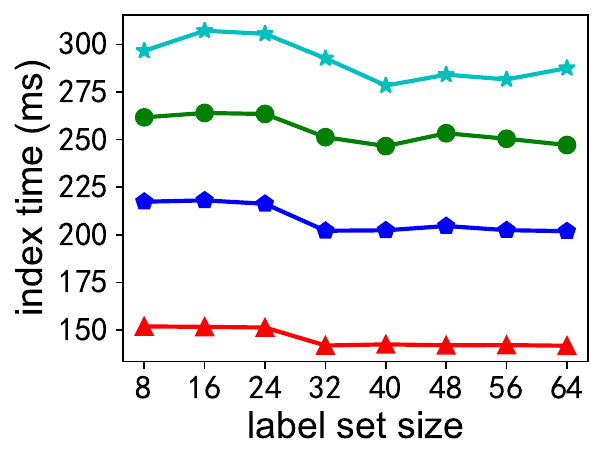}
        \caption{Index time (ms)}
        \label{exp:ER-IT}
    \end{subfigure}
    \begin{subfigure}{0.195\linewidth}
        \centering
    	\includegraphics[width=\linewidth]{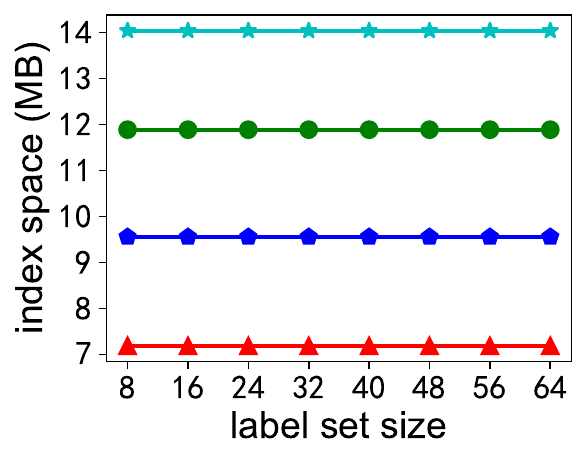}
        \caption{Index space (MB)}
        \label{exp:ER-IS}
    \end{subfigure}
    \begin{subfigure}{0.195\linewidth}
        \centering
    	\includegraphics[width=\linewidth]{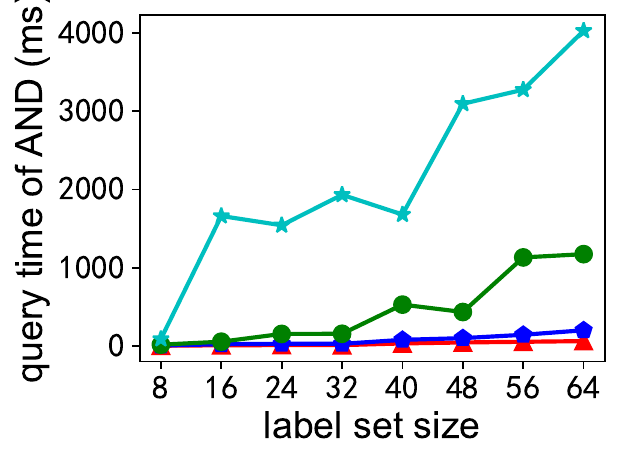}
        \caption{Time of $\mathbb{AND}$-queries (ms)}
        \label{exp:ER-AND}
    \end{subfigure}
    \begin{subfigure}{0.195\linewidth}
        \centering
    	\includegraphics[width=\linewidth]{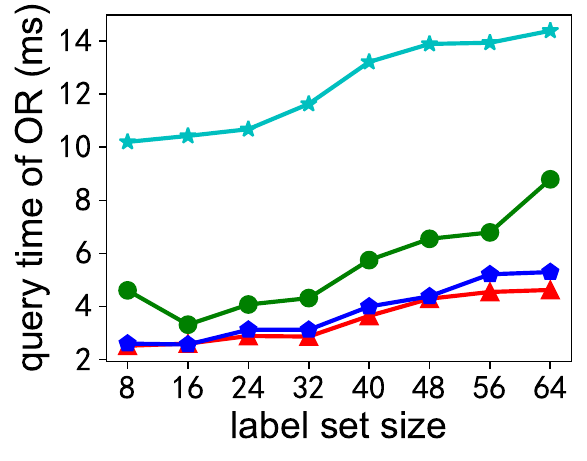}
        \caption{Time of $\mathbb{OR}$-queries (ms)}
        \label{exp:ER-OR}
    \end{subfigure}
    \begin{subfigure}{0.195\linewidth}
        \centering
    	\includegraphics[width=\linewidth]{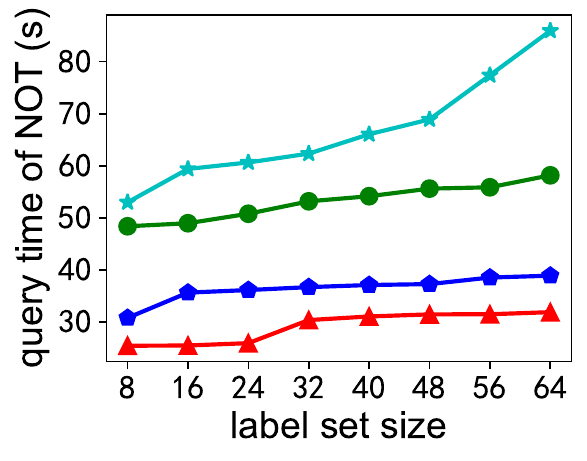}
        \caption{Time of $\mathbb{NOT}$-queries (s)}
        \label{exp:ER-NOT}
    \end{subfigure}
    \vspace{-10pt}
    \caption{\small{Indexing time, index space and execution time of $\mathbb{AND}$-, $\mathbb{OR}$- and $\mathbb{NOT}$-queries for ER-datasets with $|V|=200k$}}
    \label{exp:ER}
    \vspace{-4mm}
\end{figure*}

\begin{figure*}[htbp]
    \centering
    \begin{subfigure}{0.6\linewidth}
        \centering
        \includegraphics[width=0.95\linewidth]{images/legend.pdf}
    \end{subfigure}
    \vspace{-3mm}
        
    \begin{subfigure}{0.19\linewidth}
	\centering
    	\includegraphics[width=\linewidth]{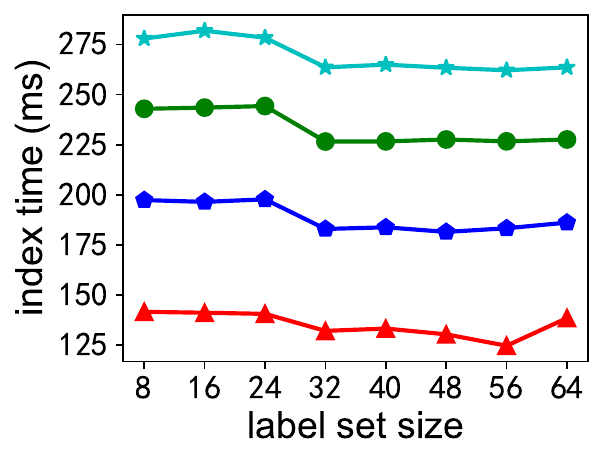}
        \vspace{-3mm}
        \caption{Index time (ms)}
        \label{exp:PA-IT}
    \end{subfigure}
    \begin{subfigure}{0.195\linewidth}
        \centering
    	\includegraphics[width=\linewidth]{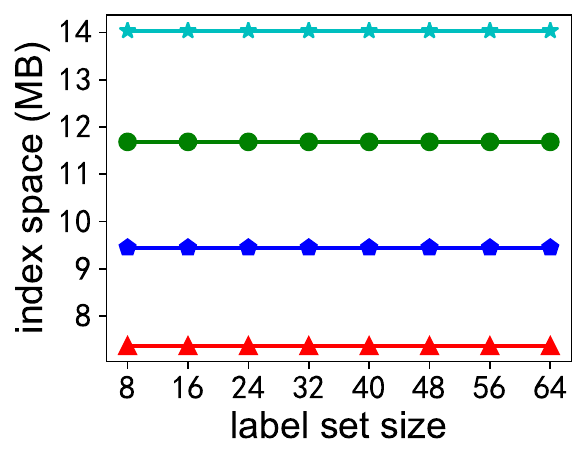}
        \vspace{-3mm}
        \caption{Index space (MB)}
        \label{exp:PA-IS}
    \end{subfigure}
    \begin{subfigure}{0.195\linewidth}
        \centering
    	\includegraphics[width=\linewidth]{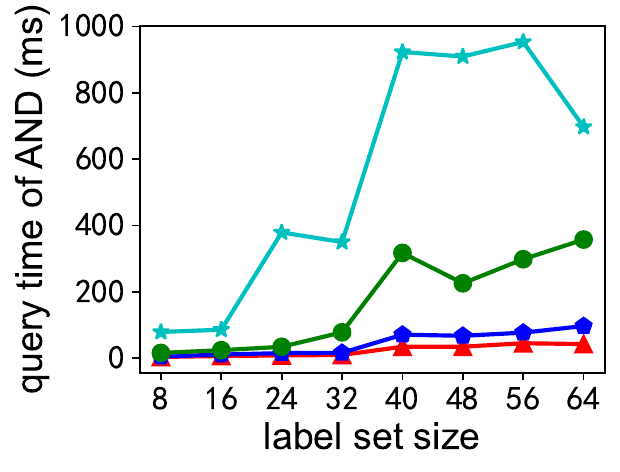}
        \vspace{-3mm}
        \caption{Time of $\mathbb{AND}$-queries (ms)}
        \label{exp:PA-AND}
    \end{subfigure}
    \begin{subfigure}{0.195\linewidth}
        \centering
    	\includegraphics[width=\linewidth]{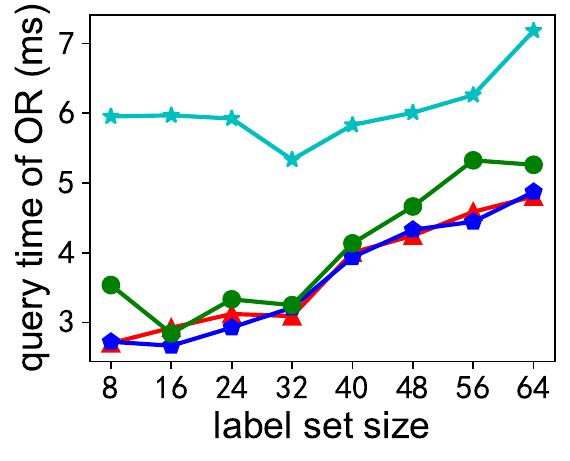}
        \vspace{-3mm}
        \caption{Time of $\mathbb{OR}$-queries (ms)}
        \label{exp:PA-OR}
    \end{subfigure}
    \begin{subfigure}{0.195\linewidth}
        \centering
    	\includegraphics[width=\linewidth]{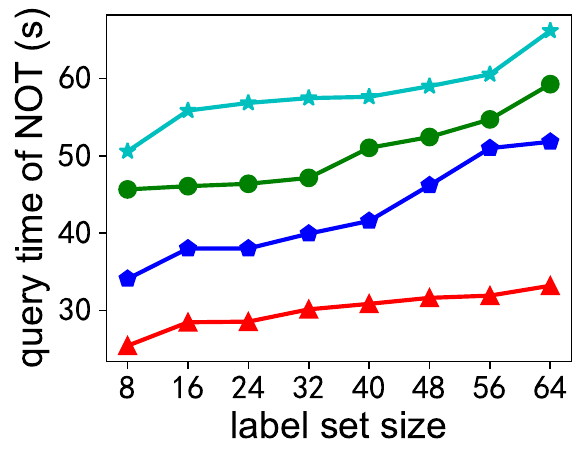}
        \vspace{-3mm}
        \caption{Time of $\mathbb{NOT}$-queries (s)}
        \label{exp:PA-NOT}
    \end{subfigure}
    \vspace{-7pt}
    \caption{\small{Indexing time, index space and execution time of $\mathbb{AND}$-, $\mathbb{OR}$- and $\mathbb{NOT}$-queries for PA-datasets with $|V|=200k$}}
    \label{exp:PA}
    \vspace{-2mm}
\end{figure*}

\textbf{Index Space:} According to the data of index space in Table \ref{tab:index}, TDR occupies one to three orders of magnitude less space compared to P2H+ and PDU for the given datasets. The reason is that TDR utilizes hash arrays storing reachable vertices and labels, while both P2H+ and PDU maintain full reachable indexes for all vertices. The full indexes of P2H+ require more space than the hash array of TDR. PDU specifically handles vertices with a degree of 1, leading to a reduced index space relative to P2H+. 

\vspace{-5pt}
\subsection{Answering PCR Queries}\label{exp:query}

Here, we first evaluate the performance of TDR using three query sets in which the operators are $\mathbb{AND}$, $\mathbb{OR}$, or $\mathbb{NOT}$, respectively. P2H+ and PDU can only process label-constrained queries that can be described using pattern-constrained reachability queries (Section \ref{sec:preliminary}). Thus, we translate LCR-queries into PCR-queries and then compare our method with P2H+ and PDU.


\begin{table}[htb]
\centering
\caption{\small{The execution time of LCR query on real digraphs (in millisecond). Here, "-" indicates that building index fails.}}
\label{tab:LCR}
\resizebox{0.95\linewidth}{!}{
\begin{tabular}{l|rrr|rrr}
\hline
\multirow{2}{*}{\textbf{Dataset}} & \multicolumn{3}{c|}{\textbf{true-query}}&\multicolumn{3}{c}{\textbf{false-query}} \\
\cline{2-7}
~ & P2H+ & PDU & TDR & P2H+ & PDU & TDR \\
\hline 

Youtube	&	0.98	&	0.61	&	0.44s	&	1.23	&	0.94	&	0.13s	\\
StringFC	&	14.7	&	1.35	&	490	&	20.21	&	6.5	&	16.15	\\
email	&	0.88	&	0.75	&	98.72	&	0.7	&	0.31	&	9.22	\\
webStanford	&	-	&	1.22	&	30.61	&	-	&	0.16	&	3.43	\\
NotreDame	&	3.35	&	13.53	&	2.85	&	3.53	&	3.66	&	1.77	\\
citeseer	&	-	&	0.21s	&	5.34	&	-	&	25.53	&	2.03	\\
webBerkStan	&	-	&	1.05	&	32.15	&	-	&	0.28	&	5.19	\\
wikitalk	&	-	&	-	&	16.04	&	-	&	-	&	6.43	\\
socPokecL	&	-	&	-	&	266.21s	&	-	&	-	&	51.70s	\\
twitter	&	-	&	-	&	3042.96s	&	-	&	-	&	163.26s	\\

\hline
\end{tabular}
}
\end{table}

\textbf{$\mathbb{AND}$, $\mathbb{OR}$, and $\mathbb{NOT}$-queries}. The execution time of TDR and DFS is shown in Table \ref{tab:PCR}. TDR is significantly faster than DFS across all datasets for three query sets, intuitively demonstrating the efficiency of TDR indexes. 
For all datasets, $\mathbb{NOT}$ queries execute more quickly on the false-query set as compared to the true-query set. This observation, however, does not extend to $\mathbb{AND}$ and $\mathbb{OR}$ queries. The reason lies in the fact that to answer $\mathbb{AND}$ and $\mathbb{OR}$ queries, it is sufficient to verify the existence of a path satisfying the given label sets. In contrast, $\mathbb{NOT}$ queries require identifying a path lacking the given label set. Consequently, the search space for $\mathbb{NOT}$ queries is typically much larger compared to that for $\mathbb{AND}$ and $\mathbb{OR}$ queries. 
Similarly, the query time of $\mathbb{OR}$ queries is generally smaller compared to $\mathbb{AND}$ queries, the constraints of $\mathbb{OR}$ queries are less stringent, making them faster. It is also true when answering a false-query because our method needs to ensure that all paths fail to meet the constraint specified in the $\mathbb{AND}$ query or $\mathbb{OR}$ query.


\textbf{LCR-queries}. 
P2H+ can only load 4 of 10 data sets, while PDU can manage to load 7. The reason is that it is costly for P2H+ and TDR to build their indices. Therefore, we evaluate our approach in comparison with the two competitors on the datasets they successfully load. The execution time (Table \ref{tab:LCR}) shows that TDR only outperforms P2H+ and PDU on \textit{citeseer} and \textit{NotreDame}. One reason is that P2H+ and PDU construct comprehensive indices with detailed reachability information. Another significant reason is that LCR queries specify label sets inclusive of all labels present along solution paths, thereby substantially reducing the search space. Nevertheless, any LCR query is translated into a PCR query, which is a combination of $\mathbb{AND}$ sub-queries and $\mathbb{NOT}$ sub-queries, each specifying a subset of labels on solution paths. Consequently, the search space of the PCR query is notably larger  than that of the initial LCR query. 
However, TDR still outperforms P2H+ and PDU on \textit{citeseer} and \textit{NotreDame}. For each dataset, our method answers false-queries faster than true-queries because TDR index is designed for answering false queries. Thus, TDR can efficiently answer real reachability queries on large graphs (e.g., \textit{socPokecL}, \textit{twitter}) because they are sparse and most of their vertex pairs tend to be unreachable.  

\subsection{Impact of Graph Characteristics}\label{exp:DandL}
To show the impact of graph characteristics, we run experiments on ER-datasets and PA-datasets. We vary the average degree $D$ $(2-8)$, and the size of the label set $|\zeta|$ $(8-64)$. The query time is the average of true-queries and false-queries.

\textbf{Indexing Time:} 
As illustrated in \figurename~\ref{exp:ER-IT} and \ref{exp:PA-IT}, the number of labels $|\zeta|$ has minimal impact on our indexing time for a given $D$ of two data sets because each label is independently hashed into a bit mask. Moreover, with a fixed vertex number, an increase in $D$ will result in more edges, thereby producing additional groups in the index and leading to larger indexing time. Still, for both ER-graphs and PA-graphs, \figurename~\ref{exp:ER-IT} and \ref{exp:PA-IT} demonstrate a linear increase in indexing time (indicated by the evenly spaced polylines) as $D$ grows. 

\textbf{Index Space:} 
The number of hashing groups for the horizontal dimension rises with the vertex's out-degree, which is intimately connected to the average degree $D$. Consequently, as depicted in \figurename~\ref{exp:ER-IS} and \ref{exp:PA-IS}, a higher average degree $D$ leads to a larger index space for both ER-graphs and PA-graphs. However, since label sets in each way are mapped into fixed-size bit masks, the size of the label set $|\zeta|$ will not greatly affect the index space. 

\textbf{Query Time:} For both ER and PA graphs, the query time of $\mathbb{AND}$ queries and $\mathbb{OR}$ queries grows when a specific number of labels reaches (\figurename~\ref{exp:ER-AND} - \ref{exp:ER-NOT} and \figurename~\ref{exp:PA-AND} - \ref{exp:PA-NOT}). Due to the edge labels being hashed into bit masks of fixed size, hash collisions become unavoidable as the label number grows, influencing query performance. For true-queries, answering $\mathbb{AND}$ and $\mathbb{OR}$ queries can rapidly yield a true result if the label set are matched. Conversely, for false-queries, every potential path must be examined. Thus, the number of potential paths also grows as $D$ increases. As a result, query times for $\mathbb{AND}$ and $\mathbb{OR}$ queries increase concurrently with $D$. For $\mathbb{NOT}$ query where none of the labels in the query is allowed to be present in any solution path, as the value of $D$ increases, it needs to traverse an increased number of paths to answer $\mathbb{NOT}$ query. Similarly, with the growth of $|\zeta|$, more labels can appear in the paths, leading to the examination of more candidate paths. Consequently, an increase in either $D$ or $|\zeta|$ generally results in longer query times, as demonstrated in \figurename~\ref{exp:ER-NOT} and \ref{exp:PA-NOT}. 

\section{Conclusions}\label{sec:conclusion}

In the paper, we initially define pattern-constrained reachability queries, which allow complex patterns in reachability queries on multi-label graphs. To efficiently address these queries, we divide the reachable vertices of a given vertex into distinct groups where the paths sharing common vertices are placed together. Subsequently, we construct a two-dimensional reachability index for each vertex by indexing the horizontal and vertical projections of each group, respectively. The multi-way index (horizontal dimension) and path index (vertical dimension) enable both broad-range and close-range pruning. This two-dimensional reachability index allows us to avoid unnecessary searches when answering PCR queries. Our experimental results on 10 real graphs demonstrate that the index is significantly smaller than the state-of-the-art LCR indexing methods, providing an efficient solution for answering pattern-constrained reachability queries. 


\bibliographystyle{IEEEtran}
\bibliography{reference}

\begin{thebibliography}{10}
\providecommand{\url}[1]{#1}
\csname url@samestyle\endcsname
\providecommand{\newblock}{\relax}
\providecommand{\bibinfo}[2]{#2}
\providecommand{\BIBentrySTDinterwordspacing}{\spaceskip=0pt\relax}
\providecommand{\BIBentryALTinterwordstretchfactor}{4}
\providecommand{\BIBentryALTinterwordspacing}{\spaceskip=\fontdimen2\font plus
\BIBentryALTinterwordstretchfactor\fontdimen3\font minus
  \fontdimen4\font\relax}
\providecommand{\BIBforeignlanguage}[2]{{%
\expandafter\ifx\csname l@#1\endcsname\relax
\typeout{** WARNING: IEEEtran.bst: No hyphenation pattern has been}%
\typeout{** loaded for the language `#1'. Using the pattern for}%
\typeout{** the default language instead.}%
\else
\language=\csname l@#1\endcsname
\fi
#2}}
\providecommand{\BIBdecl}{\relax}
\BIBdecl

\bibitem{DRS}
\BIBentryALTinterwordspacing
K.~Macropol and A.~Singh, ``Reachability analysis and modeling of dynamic event
  networks,'' in \emph{Proceedings of the 2012th European Conference on Machine
  Learning and Knowledge Discovery in Databases - Volume Part I}, ser.
  ECMLPKDD'12.\hskip 1em plus 0.5em minus 0.4em\relax Berlin, Heidelberg:
  Springer-Verlag, 2012, p. 442–457. [Online]. Available:
  \url{https://doi.org/10.1007/978-3-642-33460-3$\_$34}
\BIBentrySTDinterwordspacing

\bibitem{SOIT}
\BIBentryALTinterwordspacing
M.-F. Chiang, Y.-H. Lin, W.-C. Peng, and P.~S. Yu, ``Inferring distant-time
  location in low-sampling-rate trajectories,'' in \emph{Proceedings of the
  19th ACM SIGKDD International Conference on Knowledge Discovery and Data
  Mining}, ser. KDD '13.\hskip 1em plus 0.5em minus 0.4em\relax New York, NY,
  USA: Association for Computing Machinery, 2013, p. 1454–1457. [Online].
  Available: \url{https://doi.org/10.1145/2487575.2487707}
\BIBentrySTDinterwordspacing

\bibitem{KIM2017217}
\BIBentryALTinterwordspacing
D.~Kim, D.~Hyeon, J.~Oh, W.-S. Han, and H.~Yu, ``Influence maximization based
  on reachability sketches in dynamic graphs,'' \emph{Information Sciences},
  vol. 394-395, pp. 217--231, 2017. [Online]. Available:
  \url{https://www.sciencedirect.com/science/article/pii/S0020025517305121}
\BIBentrySTDinterwordspacing

\bibitem{2006'Dual}
\BIBentryALTinterwordspacing
H.~Wang, H.~He, J.~Yang, P.~S. Yu, and J.~X. Yu, ``Dual labeling: Answering
  graph reachability queries in constant time,'' in \emph{Proceedings of the
  22nd International Conference on Data Engineering, {ICDE} 2006, 3-8 April
  2006, Atlanta, GA, {USA}}, L.~Liu, A.~Reuter, K.~Whang, and J.~Zhang,
  Eds.\hskip 1em plus 0.5em minus 0.4em\relax Atlanta, GA: {IEEE} Computer
  Society, 2006, p.~75. [Online]. Available:
  \url{https://doi.org/10.1109/ICDE.2006.53}
\BIBentrySTDinterwordspacing

\bibitem{2009'3-hop}
\BIBentryALTinterwordspacing
R.~Jin, Y.~Xiang, N.~Ruan, and D.~Fuhry, ``3-hop: a high-compression indexing
  scheme for reachability query,'' in \emph{Proceedings of the {ACM} {SIGMOD}
  International Conference on Management of Data, {SIGMOD} 2009, Providence,
  Rhode Island, USA, June 29 - July 2, 2009}.\hskip 1em plus 0.5em minus
  0.4em\relax Providence, Rhode Island, USA: {ACM}, 2009, pp. 813--826.
  [Online]. Available: \url{https://doi.org/10.1145/1559845.1559930}
\BIBentrySTDinterwordspacing

\bibitem{2011'Path-tree}
\BIBentryALTinterwordspacing
R.~Jin, N.~Ruan, Y.~Xiang, and H.~Wang, ``Path-tree: An efficient reachability
  indexing scheme for large directed graphs,'' \emph{{ACM} Trans. Database
  Syst.}, vol.~36, no.~1, pp. 7:1--7:44, 2011. [Online]. Available:
  \url{https://doi.org/10.1145/1929934.1929941}
\BIBentrySTDinterwordspacing

\bibitem{chen'2008'chain}
\BIBentryALTinterwordspacing
Y.~Chen and Y.~Chen, ``An efficient algorithm for answering graph reachability
  queries,'' in \emph{Proceedings of the 24th International Conference on Data
  Engineering, {ICDE} 2008, April 7-12, 2008, Canc{\'{u}}n, Mexico}, G.~Alonso,
  J.~A. Blakeley, and A.~L.~P. Chen, Eds.\hskip 1em plus 0.5em minus
  0.4em\relax Canc{\'{u}}n, Mexico: {IEEE} Computer Society, 2008, pp.
  893--902. [Online]. Available:
  \url{https://doi.org/10.1109/ICDE.2008.4497498}
\BIBentrySTDinterwordspacing

\bibitem{2005'HLSS}
\BIBentryALTinterwordspacing
H.~He, H.~Wang, J.~Yang, and P.~S. Yu, ``Compact reachability labeling for
  graph-structured data,'' in \emph{Proceedings of the 2005 {ACM} {CIKM}
  International Conference on Information and Knowledge Management, Bremen,
  Germany, October 31 - November 5, 2005}.\hskip 1em plus 0.5em minus
  0.4em\relax Bremen, Germany: {ACM}, 2005, pp. 594--601. [Online]. Available:
  \url{https://doi.org/10.1145/1099554.1099708}
\BIBentrySTDinterwordspacing

\bibitem{2005'HOPI}
\BIBentryALTinterwordspacing
R.~Schenkel, A.~Theobald, and G.~Weikum, ``Efficient creation and incremental
  maintenance of the {HOPI} index for complex {XML} document collections,'' in
  \emph{Proceedings of the 21st International Conference on Data Engineering,
  {ICDE} 2005, 5-8 April 2005, Tokyo, Japan}.\hskip 1em plus 0.5em minus
  0.4em\relax Tokyo, Japan: {IEEE} Computer Society, 2005, pp. 360--371.
  [Online]. Available: \url{https://doi.org/10.1109/ICDE.2005.57}
\BIBentrySTDinterwordspacing

\bibitem{2010'GRAIL}
\BIBentryALTinterwordspacing
H.~Yildirim, V.~Chaoji, and M.~J. Zaki, ``{GRAIL:} scalable reachability index
  for large graphs,'' \emph{Proc. {VLDB} Endow.}, vol.~3, no.~1, pp. 276--284,
  2010. [Online]. Available:
  \url{http://www.vldb.org/pvldb/vldb2010/pvldb$\_$vol3/R24.pdf}
\BIBentrySTDinterwordspacing

\bibitem{2012'GRAIL}
\BIBentryALTinterwordspacing
------, ``{GRAIL:} a scalable index for reachability queries in very large
  graphs,'' \emph{{VLDB} J.}, vol.~21, no.~4, pp. 509--534, 2012. [Online].
  Available: \url{https://doi.org/10.1007/s00778-011-0256-4}
\BIBentrySTDinterwordspacing

\bibitem{2007'GRIPP}
\BIBentryALTinterwordspacing
S.~Tri{\ss}l and U.~Leser, ``Fast and practical indexing and querying of very
  large graphs,'' in \emph{Proceedings of the {ACM} {SIGMOD} International
  Conference on Management of Data, Beijing, China, June 12-14, 2007}.\hskip
  1em plus 0.5em minus 0.4em\relax Beijing, China: {ACM}, 2007, pp. 845--856.
  [Online]. Available: \url{https://doi.org/10.1145/1247480.1247573}
\BIBentrySTDinterwordspacing

\bibitem{2013'ferrari}
\BIBentryALTinterwordspacing
S.~Seufert, A.~Anand, S.~J. Bedathur, and G.~Weikum, ``{FERRARI:} flexible and
  efficient reachability range assignment for graph indexing,'' in \emph{29th
  {IEEE} International Conference on Data Engineering, {ICDE} 2013, Brisbane,
  Australia, April 8-12, 2013}.\hskip 1em plus 0.5em minus 0.4em\relax
  Brisbane, Australia: {IEEE} Computer Society, 2013, pp. 1009--1020. [Online].
  Available: \url{https://doi.org/10.1109/ICDE.2013.6544893}
\BIBentrySTDinterwordspacing

\bibitem{2014'IP}
\BIBentryALTinterwordspacing
H.~Wei, J.~X. Yu, C.~Lu, and R.~Jin, ``Reachability querying: An independent
  permutation labeling approach,'' \emph{Proc. {VLDB} Endow.}, vol.~7, no.~12,
  pp. 1191--1202, 2014. [Online]. Available:
  \url{http://www.vldb.org/pvldb/vol7/p1191-wei.pdf}
\BIBentrySTDinterwordspacing

\bibitem{2018'IP}
\BIBentryALTinterwordspacing
------, ``Reachability querying: An independent permutation labeling
  approach,'' \emph{The VLDB Journal}, vol.~27, no.~1, pp. 1--26, Feb. 2018.
  [Online]. Available: \url{https://doi.org/10.1007/s00778-017-0468-3}
\BIBentrySTDinterwordspacing

\bibitem{2013'TFLabel}
\BIBentryALTinterwordspacing
J.~Cheng, S.~Huang, H.~Wu, and A.~W. Fu, ``Tf-label: a topological-folding
  labeling scheme for reachability querying in a large graph,'' in
  \emph{Proceedings of the {ACM} {SIGMOD} International Conference on
  Management of Data, {SIGMOD} 2013, New York, NY, USA, June 22-27,
  2013}.\hskip 1em plus 0.5em minus 0.4em\relax New York, NY, USA: {ACM}, 2013,
  pp. 193--204. [Online]. Available:
  \url{https://doi.org/10.1145/2463676.2465286}
\BIBentrySTDinterwordspacing

\bibitem{1989'TCC}
\BIBentryALTinterwordspacing
R.~Agrawal, A.~Borgida, and H.~V. Jagadish, ``Efficient management of
  transitive relationships in large data and knowledge bases,'' in
  \emph{Proceedings of the 1989 {ACM} {SIGMOD} International Conference on
  Management of Data, Portland, Oregon, USA, May 31 - June 2, 1989}.\hskip 1em
  plus 0.5em minus 0.4em\relax Portland, Oregon, USA: {ACM} Press, 1989, pp.
  253--262. [Online]. Available: \url{https://doi.org/10.1145/67544.66950}
\BIBentrySTDinterwordspacing

\bibitem{2017'BFL}
\BIBentryALTinterwordspacing
J.~Su, Q.~Zhu, H.~Wei, and J.~X. Yu, ``Reachability querying: Can it be even
  faster?'' \emph{{IEEE} Trans. Knowl. Data Eng.}, vol.~29, no.~3, pp.
  683--697, 2017. [Online]. Available:
  \url{https://doi.org/10.1109/TKDE.2016.2631160}
\BIBentrySTDinterwordspacing

\bibitem{Fan'2011}
\BIBentryALTinterwordspacing
W.~Fan, J.~Li, S.~Ma, N.~Tang, and Y.~Wu, ``Adding regular expressions to graph
  reachability and pattern queries,'' in \emph{Proceedings of the 27th
  International Conference on Data Engineering, {ICDE} 2011, April 11-16, 2011,
  Hannover, Germany}.\hskip 1em plus 0.5em minus 0.4em\relax Hannover, Germany:
  {IEEE} Computer Society, 2011, pp. 39--50. [Online]. Available:
  \url{https://doi.org/10.1109/ICDE.2011.5767858}
\BIBentrySTDinterwordspacing

\bibitem{2017'LI}
\BIBentryALTinterwordspacing
L.~D.~J. Valstar, G.~H.~L. Fletcher, and Y.~Yoshida, ``Landmark indexing for
  evaluation of label-constrained reachability queries,'' in \emph{Proceedings
  of the 2017 {ACM} International Conference on Management of Data, {SIGMOD}
  Conference 2017, Chicago, IL, USA, May 14-19, 2017}.\hskip 1em plus 0.5em
  minus 0.4em\relax Chicago, IL, USA: {ACM}, 2017, pp. 345--358. [Online].
  Available: \url{https://doi.org/10.1145/3035918.3035955}
\BIBentrySTDinterwordspacing

\bibitem{2020'P2H}
\BIBentryALTinterwordspacing
Y.~Peng, Y.~Zhang, X.~Lin, L.~Qin, and W.~Zhang, ``Answering billion-scale
  label-constrained reachability queries within microsecond,'' \emph{Proc.
  {VLDB} Endow.}, vol.~13, no.~6, pp. 812--825, 2020. [Online]. Available:
  \url{http://www.vldb.org/pvldb/vol13/p812-peng.pdf}
\BIBentrySTDinterwordspacing

\bibitem{2023'PDU}
\BIBentryALTinterwordspacing
Y.~Cai and W.~Zheng, ``Answering label-constrained reachability queries via
  reduction techniques,'' in \emph{Database Systems for Advanced Applications -
  28th International Conference, {DASFAA} 2023, Tianjin, China, April 17-20,
  2023, Proceedings, Part {I}}, ser. Lecture Notes in Computer Science, vol.
  13943.\hskip 1em plus 0.5em minus 0.4em\relax Tianjin, China: Springer, 2023,
  pp. 114--131. [Online]. Available:
  \url{https://doi.org/10.1007/978-3-031-30637-2$\_$8}
\BIBentrySTDinterwordspacing

\bibitem{2022'P2H}
\BIBentryALTinterwordspacing
Y.~Peng, X.~Lin, Y.~Zhang, W.~Zhang, and L.~Qin, ``Answering reachability and
  k-reach queries on large graphs with label constraints,'' \emph{{VLDB} J.},
  vol.~31, no.~1, pp. 101--127, 2022. [Online]. Available:
  \url{https://doi.org/10.1007/s00778-021-00695-0}
\BIBentrySTDinterwordspacing

\bibitem{2021'RQuBE}
\BIBentryALTinterwordspacing
K.~Chauhan, K.~Jain, S.~Ranu, S.~Bedathur, and A.~Bagchi, ``Answering regular
  path queries through exemplars,'' \emph{Proc. {VLDB} Endow.}, vol.~15, no.~2,
  pp. 299--311, 2021. [Online]. Available:
  \url{http://www.vldb.org/pvldb/vol15/p299-ranu.pdf}
\BIBentrySTDinterwordspacing

\bibitem{2012'RPQonLG}
\BIBentryALTinterwordspacing
A.~Koschmieder and U.~Leser, ``Regular path queries on large graphs,'' in
  \emph{Scientific and Statistical Database Management - 24th International
  Conference, {SSDBM} 2012, Chania, Crete, Greece, June 25-27, 2012.
  Proceedings}, ser. Lecture Notes in Computer Science, vol. 7338.\hskip 1em
  plus 0.5em minus 0.4em\relax Chania, Crete, Greece: Springer, 2012, pp.
  177--194. [Online]. Available:
  \url{https://doi.org/10.1007/978-3-642-31235-9$\_$12}
\BIBentrySTDinterwordspacing

\bibitem{2019'ARRIVAL}
\BIBentryALTinterwordspacing
S.~Wadhwa, A.~Prasad, S.~Ranu, A.~Bagchi, and S.~Bedathur, ``Efficiently
  answering regular simple path queries on large labeled networks,'' in
  \emph{Proceedings of the 2019 International Conference on Management of Data,
  {SIGMOD} Conference 2019, Amsterdam, The Netherlands, June 30 - July 5,
  2019}.\hskip 1em plus 0.5em minus 0.4em\relax Amsterdam, The Netherlands:
  {ACM}, 2019, pp. 1463--1480. [Online]. Available:
  \url{https://doi.org/10.1145/3299869.3319882}
\BIBentrySTDinterwordspacing

\bibitem{2022'Ring}
\BIBentryALTinterwordspacing
D.~Arroyuelo, A.~Hogan, G.~Navarro, and J.~Rojas{-}Ledesma, ``Time- and
  space-efficient regular path queries,'' in \emph{38th {IEEE} International
  Conference on Data Engineering, {ICDE} 2022, Kuala Lumpur, Malaysia, May
  9-12, 2022}.\hskip 1em plus 0.5em minus 0.4em\relax Kuala Lumpur, Malaysia:
  {IEEE}, 2022, pp. 3091--3105. [Online]. Available:
  \url{https://doi.org/10.1109/ICDE53745.2022.00277}
\BIBentrySTDinterwordspacing

\bibitem{2020'Streaming}
\BIBentryALTinterwordspacing
A.~Pacaci, A.~Bonifati, and M.~T. {\"{O}}zsu, ``Regular path query evaluation
  on streaming graphs,'' in \emph{Proceedings of the 2020 International
  Conference on Management of Data, {SIGMOD} Conference 2020, online conference
  [Portland, OR, USA], June 14-19, 2020}.\hskip 1em plus 0.5em minus
  0.4em\relax Portland, OR, USA: {ACM}, 2020, pp. 1415--1430. [Online].
  Available: \url{https://doi.org/10.1145/3318464.3389733}
\BIBentrySTDinterwordspacing

\bibitem{2023'RLC}
\BIBentryALTinterwordspacing
C.~Zhang, A.~Bonifati, H.~Kapp, V.~I. Haprian, and J.~Lozi, ``A reachability
  index for recursive label-concatenated graph queries,'' in \emph{39th {IEEE}
  International Conference on Data Engineering, {ICDE} 2023, Anaheim, CA, USA,
  April 3-7, 2023}.\hskip 1em plus 0.5em minus 0.4em\relax Anaheim, CA, USA:
  {IEEE}, 2023, pp. 67--81. [Online]. Available:
  \url{https://doi.org/10.1109/ICDE55515.2023.00013}
\BIBentrySTDinterwordspacing

\bibitem{Chen'2021'Recurve}
\BIBentryALTinterwordspacing
Y.~Chen and G.~Singh, ``Graph indexing for efficient evaluation of
  label-constrained reachability queries,'' \emph{{ACM} Trans. Database Syst.},
  vol.~46, no.~2, pp. 8:1--8:50, 2021. [Online]. Available:
  \url{https://doi.org/10.1145/3451159}
\BIBentrySTDinterwordspacing

\bibitem{Vertigo}
\BIBentryALTinterwordspacing
M.~Nol\'{e} and C.~Sartiani, ``Regular path queries on massive graphs,'' in
  \emph{Proceedings of the 28th International Conference on Scientific and
  Statistical Database Management}, ser. SSDBM '16.\hskip 1em plus 0.5em minus
  0.4em\relax New York, NY, USA: Association for Computing Machinery, 2016.
  [Online]. Available: \url{https://doi.org/10.1145/2949689.2949711}
\BIBentrySTDinterwordspacing

\bibitem{RPQExample}
\BIBentryALTinterwordspacing
K.~Chauhan, K.~Jain, S.~Ranu, S.~Bedathur, and A.~Bagchi, ``Answering regular
  path queries through exemplars,'' \emph{Proc. VLDB Endow.}, vol.~15, no.~2,
  p. 299–311, Oct. 2021. [Online]. Available:
  \url{https://doi.org/10.14778/3489496.3489510}
\BIBentrySTDinterwordspacing

\bibitem{SNAP}
J.~Leskovec and R.~Sosi{\v{c}}, ``Snap: A general-purpose network analysis and
  graph-mining library,'' \emph{ACM Transactions on Intelligent Systems and
  Technology (TIST)}, vol.~8, no.~1, p.~1, 2016.

\bibitem{KONECT}
\BIBentryALTinterwordspacing
J.~Kunegis, ``{KONECT:} the koblenz network collection,'' in \emph{22nd
  International World Wide Web Conference, {WWW} '13}, Rio de Janeiro, Brazil,
  2013, pp. 1343--1350. [Online]. Available:
  \url{https://doi.org/10.1145/2487788.2488173}
\BIBentrySTDinterwordspacing

\bibitem{1999'PA}
A.~L. Barabási and R.~Albert, ``Emergence of scaling in random networks,''
  \emph{Science}, vol. 286, no. 5439, pp. 509--512, 1999.

\bibitem{1959'ER}
P.~Erdős and A.~Rényi, ``On random graphs,'' \emph{Publicationes
  Mathematicae}, vol.~6, pp. 290--297, 1959.

\bibitem{SAT'Cook71}
\BIBentryALTinterwordspacing
S.~A. Cook, ``The complexity of theorem-proving procedures,'' in
  \emph{Proceedings of the 3rd Annual {ACM} Symposium on Theory of Computing,
  May 3-5, 1971, Shaker Heights, Ohio, {USA}}, M.~A. Harrison, R.~B. Banerji,
  and J.~D. Ullman, Eds.\hskip 1em plus 0.5em minus 0.4em\relax {ACM}, 1971,
  pp. 151--158. [Online]. Available:
  \url{https://doi.org/10.1145/800157.805047}
\BIBentrySTDinterwordspacing

\bibitem{2023'overview}
\BIBentryALTinterwordspacing
C.~Zhang, A.~Bonifati, and M.~T. {\"{O}}zsu, ``An overview of reachability
  indexes on graphs,'' in \emph{Companion of the 2023 International Conference
  on Management of Data, {SIGMOD/PODS} 2023, Seattle, WA, USA, June 18-23,
  2023}.\hskip 1em plus 0.5em minus 0.4em\relax {ACM}, 2023, pp. 61--68.
  [Online]. Available: \url{https://doi.org/10.1145/3555041.3589408}
\BIBentrySTDinterwordspacing

\bibitem{2018'MGTag}
\BIBentryALTinterwordspacing
S.~Zhou, P.~Yuan, L.~Liu, and H.~Jin, ``Mgtag: a multi-dimensional graph
  labeling scheme for fast reachability queries,'' in \emph{34th {IEEE}
  International Conference on Data Engineering, {ICDE} 2018, Paris, France,
  April 16-19, 2018}.\hskip 1em plus 0.5em minus 0.4em\relax Paris, France:
  {IEEE} Computer Society, 2018, pp. 1372--1375. [Online]. Available:
  \url{https://doi.org/10.1109/ICDE.2018.00153}
\BIBentrySTDinterwordspacing

\bibitem{2022'MGTag}
P.~Yuan, Y.~You, S.~Zhou, H.~Jin, and L.~Liu, ``Providing fast reachability
  query services with mgtag: A multi-dimensional graph labeling method,''
  \emph{IEEE Transactions on Services Computing}, vol.~15, no.~2, pp.
  1000--1011, 2022.

\bibitem{2008'pathtree}
\BIBentryALTinterwordspacing
R.~Jin, Y.~Xiang, N.~Ruan, and H.~Wang, ``Efficiently answering reachability
  queries on very large directed graphs,'' in \emph{Proceedings of the {ACM}
  {SIGMOD} International Conference on Management of Data, {SIGMOD} 2008,
  Vancouver, BC, Canada, June 10-12, 2008}.\hskip 1em plus 0.5em minus
  0.4em\relax Vancouver, BC, Canada: {ACM}, 2008, pp. 595--608. [Online].
  Available: \url{https://doi.org/10.1145/1376616.1376677}
\BIBentrySTDinterwordspacing

\bibitem{2011'PWAH}
\BIBentryALTinterwordspacing
S.~J. van Schaik and O.~de~Moor, ``A memory efficient reachability data
  structure through bit vector compression,'' in \emph{Proceedings of the {ACM}
  {SIGMOD} International Conference on Management of Data, {SIGMOD} 2011,
  Athens, Greece, June 12-16, 2011}.\hskip 1em plus 0.5em minus 0.4em\relax
  Athens, Greece: {ACM}, 2011, pp. 913--924. [Online]. Available:
  \url{https://doi.org/10.1145/1989323.1989419}
\BIBentrySTDinterwordspacing

\bibitem{2014'feline}
\BIBentryALTinterwordspacing
R.~R. Veloso, L.~Cerf, W.~M. Jr., and M.~J. Zaki, ``Reachability queries in
  very large graphs: {A} fast refined online search approach,'' in
  \emph{Proceedings of the 17th International Conference on Extending Database
  Technology, {EDBT} 2014, Athens, Greece, March 24-28, 2014}.\hskip 1em plus
  0.5em minus 0.4em\relax Athens, Greece: OpenProceedings.org, 2014, pp.
  511--522. [Online]. Available: \url{https://doi.org/10.5441/002/edbt.2014.46}
\BIBentrySTDinterwordspacing

\bibitem{2006'SCI}
\BIBentryALTinterwordspacing
J.~Cheng, J.~X. Yu, X.~Lin, H.~Wang, and P.~S. Yu, ``Fast computation of
  reachability labeling for large graphs,'' in \emph{Advances in Database
  Technology - {EDBT} 2006, 10th International Conference on Extending Database
  Technology, Munich, Germany, March 26-31, 2006, Proceedings}, ser. Lecture
  Notes in Computer Science, Y.~E. Ioannidis, M.~H. Scholl, J.~W. Schmidt,
  F.~Matthes, M.~Hatzopoulos, K.~B{\"{o}}hm, A.~Kemper, T.~Grust, and
  C.~B{\"{o}}hm, Eds., vol. 3896.\hskip 1em plus 0.5em minus 0.4em\relax
  Munich, Germany: Springer, 2006, pp. 961--979. [Online]. Available:
  \url{https://doi.org/10.1007/11687238$\_$56}
\BIBentrySTDinterwordspacing

\bibitem{2010'Path-hop}
\BIBentryALTinterwordspacing
J.~Cai and C.~K. Poon, ``Path-hop: efficiently indexing large graphs for
  reachability queries,'' in \emph{Proceedings of the 19th {ACM} Conference on
  Information and Knowledge Management, {CIKM} 2010, Toronto, Ontario, Canada,
  October 26-30, 2010}.\hskip 1em plus 0.5em minus 0.4em\relax Toronto,
  Ontario, Canada: {ACM}, 2010, pp. 119--128. [Online]. Available:
  \url{https://doi.org/10.1145/1871437.1871457}
\BIBentrySTDinterwordspacing

\bibitem{2013'HL}
\BIBentryALTinterwordspacing
R.~Jin and G.~Wang, ``Simple, fast, and scalable reachability oracle,''
  \emph{Proc. {VLDB} Endow.}, vol.~6, no.~14, pp. 1978--1989, 2013. [Online].
  Available: \url{http://www.vldb.org/pvldb/vol6/p1978-jin.pdf}
\BIBentrySTDinterwordspacing

\bibitem{1995'RPQNPHard}
\BIBentryALTinterwordspacing
A.~O. Mendelzon and P.~T. Wood, ``Finding regular simple paths in graph
  databases,'' \emph{{SIAM} J. Comput.}, vol.~24, no.~6, pp. 1235--1258, 1995.
  [Online]. Available: \url{https://doi.org/10.1137/S009753979122370X}
\BIBentrySTDinterwordspacing

\end{thebibliography}

\appendix
\newpage

\section{Problem Complexity} \label{appendix:problem}
Answering an LCR query can be performed in polynomial time \cite{2019'ARRIVAL}, while addressing PCR queries is NP-hard. 
\begin{mythm}
\label{thr:complexity}
    PCR is an NP-hard problem.
\end{mythm}

\begin{proof}
Since the SAT problem is an NP-complete problem \cite{SAT'Cook71}, we can reduce SAT to PCR as follows: without loss of generality, assume a PCR query $ \xymatrix @C=1.5em {u\ar@{~>}[r]^?_{\mathcal{P}} & v}$ where $\mathcal{P}$ is a conjunction of $m$ labels ($m<\|\zeta\|$), that is, $\mathcal{P}$=$\mathbb{AND}\{l_i\}_{i<m}$. We expand $\mathcal{P}$ with new variables, each of which denotes a label $l_j$ ($m\leq j<\|\zeta\|$), i.e., $\mathcal{P}=$$\mathbb{AND}\{l_i\}_{i<m}$ $\mathbb{OR}\{l_i\}_{m\leq j<\|\zeta\|}$. Given any clause $(x_1 \land x_2 \land \cdots \land x_m)$ in SAT, we consider that clause $x_i$ is mapped to $i$-th edge of paths from $u$ to $v$. Thus, we can describe the PCR query by rewriting clause $x_i$ as $\mathcal{P}$. Each label of $\mathcal{P}$ corresponds to a variable, which has two states: presence or absence of the label on each edge of paths from source vertex $u$ to the target vertex $v$ while each variable in SAT can only take two values of 0 and 1. Therefore, PCR is NP-hard. 
  
 
 \end{proof}

\section{Complexity Analysis}
Initially, we analyze the complexity of building the index, followed by a complexity analysis of answering PCR queries using the index.
\subsection{Building the index}\label{CC}

We first analyze the complexity of the index space. For each vertex $u$, the index consists of the horizontal dimension and vertical dimension index. $u$'s successors are divided into $\frac{|Suc(u)|}{m}$ blocks according to its out-degree where $m$ is the number of neighboring vertices in each block. For each block, in horizontal dimension, the reachable vertex index and label index are stored in two bit arrays, respectively; in vertical dimension, two bit arrays of length $k$ are allocated for path index ($\mathcal{V}_{lab}$). Hence, the storage space of TDR is $|V|(k+\sum_{i=1}^{|V|}\frac{|Suc(u_i)|}{m})$. Assume that the average degree is $\bar{d}$, then the storage complexity of the index is $O((\frac{\bar{d}}{m}+k)\cdot |V|)$.

We take a bottom-up approach to process each vertex in turn, and each vertex is pushed into the stack only once during index building. When vertex $u$ is visited, we index $u$ based on its successors. Multi-way hashing assigns these successors to different groups and merges their index into $u$. Each successor is processed in O(1) time. Path hashing requires merging the last $k-1$ bits of the successors' $\mathcal{V}_{lab}$ into $u$, so each successor is visited $k-1$ times. Thus, the total time is $|V|+\sum_{i=1}^{|V|}(1+k-1)\cdot|Suc(u_i)|$, that is, the time complexity is $O(|V|+k\cdot|E|)$.
\begin{figure*}[htp]
	\centering
	
	\begin{subfigure}{0.195\linewidth}
	    \centering
    	\includegraphics[width=\linewidth]{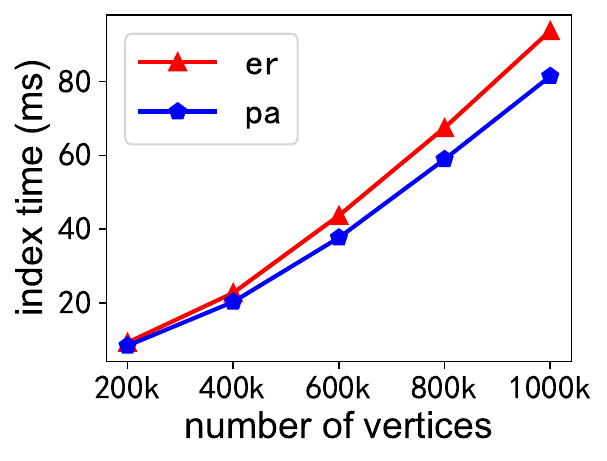}
        \caption{Index time (ms)}
        \label{exp:number-IT}
    \end{subfigure}
    \begin{subfigure}{0.2\linewidth}
        \centering
    	\includegraphics[width=\linewidth]{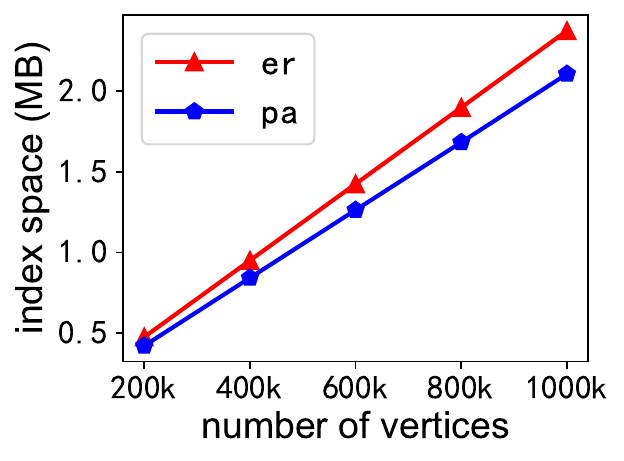}
        \caption{Index space (MB)}
        \label{exp:number-IS}
    \end{subfigure}
    \begin{subfigure}{0.19\linewidth}
        \centering
    	\includegraphics[width=\linewidth]{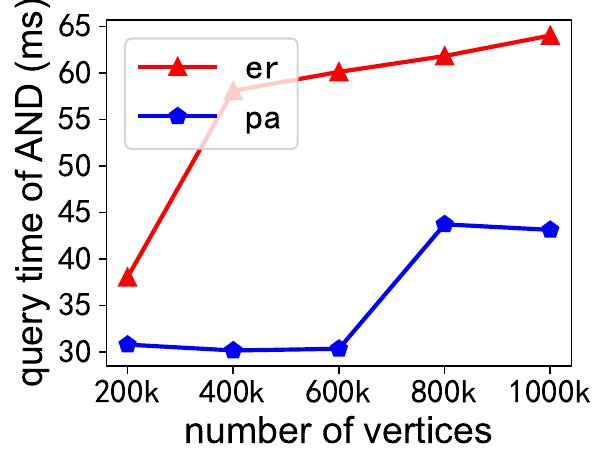}
        \caption{Time of $\mathbb{AND}$-queries (ms)}
        \label{exp:number-AND}
    \end{subfigure}
    \begin{subfigure}{0.19\linewidth}
        \centering
    	\includegraphics[width=\linewidth]{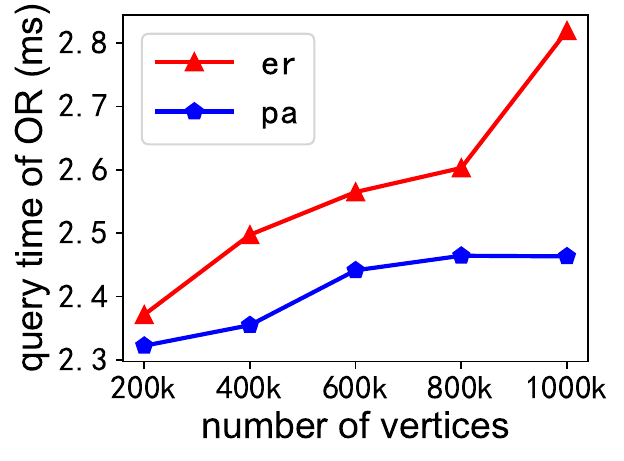}
        \caption{Time of $\mathbb{OR}$-queries (ms)}
        \label{exp:number-OR}
    \end{subfigure}
    \begin{subfigure}{0.2\linewidth}
        \centering
    	\includegraphics[width=\linewidth]{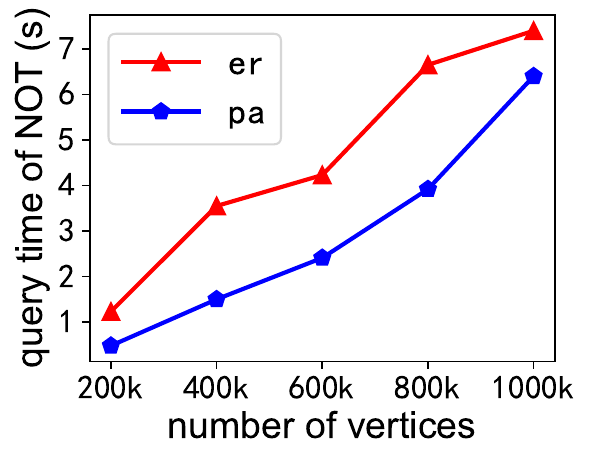}
        \caption{Time of $\mathbb{NOT}$-queries (s)}
        \label{exp:number-NOT}
    \end{subfigure}
    \caption{\small{Indexing time, index space and execution time of $\mathbb{AND}$-, $\mathbb{OR}$- and $\mathbb{NOT}$-queries for graphs with $D=2$ and $|L|=32$}}
	\label{exp:number}
\end{figure*}

\subsection{Answering reachability query}\label{sec:CA}
Given a reachability query $ \xymatrix @C=1.5em {u\ar@{~>}[r]^?_{\mathcal{P}} & v}$, 
we analyze the time complexity to answer the query. Let $N(u, v)$ denote the number of vertices visited when answering whether $u$ can reach $v$. 
When answering a reachability query, we need to examine the TDR index associated with each vertex. If $u$’s index (such as $\mathcal{I}$, $\mathcal{N}$, and $\mathcal{V}_{lab}$) shows that the query is unreachable or false, the search process will terminate and return the answer immediately. For the case, $N(u,v)=1$. Additionally, $u$’s index can return  \textit{reachability} or \textit{true} with probability $\rho$. Consequently, our approach must verify each of $u$'s neighbors one by one. Given that these neighbors are distributed into $m$ groups, suppose the probability of the target vertex being in group $i$ is $\rho_{i}$. This implies that with probability $\rho_{i}$, the successors of $u$ in group $i$ will be examined. Hence, the number of vertices to be explored is $N(u,v)$ = 1+$\rho$$\sum_{0 \leq i < m}(\rho_{i}\sum_{w \in \mathcal{H}(u)_i}N(w,v))$. For ease of discussion, we assume every vertex $w$ has a probability $\rho$ of returning \textit{reachability} or \textit{true}. Otherwise, the branches originating from $w$ can be disregarded with probability (1-$\rho$).

Let $t(u,v)$ denote the time to answer reachability query $(u,v)$. Assume $\bar{d}$ to be the average degree and $\bar{l}$ is the average length of the paths.
The above process can be written as $t(u,v) = \rho\sum_{w\in Suc(u)}(t(w,v)+ O(1))$, then we have

$~~~t(u,v) =\rho\sum_{w\in Suc(u)}(t(w,v)+ O(1))$

$ =\rho \sum_{w \in Suc(u)} (\rho \sum_{w^{'} \in Suc(w)} ( t(w^{'},v)+ O(1)))$

    $=(\rho\bar{d})^{\bar{l}}+O(1)\sum_{i=0}^{\bar{l}}(\rho\bar{d})^i$ 

$=(\rho\bar{d})^{\bar{l}}+\frac{(\rho\bar{d})^{\bar{l}+1}-1}{\rho\bar{d}-1} O(1)$

It shows that search complexity may grow exponentially with increasing search depth. If forward pruning can remove large fractions of branches without further check, then the complexity is of the same order as the complexity of the process in each vertex.

\section{Further Evaluation}
Since our approach maps the reachable vertices of a vertex into bit sets with fixed size, our approach can work on large graphs. Here, we evaluate the scalability of TDR using synthetic graphs.

\subsection{Scalability}\label{exp:structure}

To investigate the scalability of our method, we set $D=6$ and $|\zeta|=32$, and then vary the number of vertices from $200K$ to $1000K$ in both the ER-datasets and PA-datasets. The indexing time, index space, and query time of three subpatterns for both the ER-graphs and PA-graphs are presented in Fig. \ref{exp:number}, respectively.

\textbf{Index:} Since TDR stores the reachability information of each vertex, the indexing time and space required for our indices scale linearly with the number of vertices (Figs. \ref{exp:number-IT} and \ref{exp:number-IS}). Additionally, the uniform degree distribution in ER graphs results in fewer vertices with a degree of 0 compared to PA graphs. As a result, the indexing process in ER graphs involves more vertices than in PA graphs, leading to a bit more indexing time and a larger index space for ER graphs. 

\textbf{Query Time:} The execution times of $\mathbb{AND}$, $\mathbb{OR}$, and $\mathbb{NOT}$ queries grow proportionally to the number of vertices in both ER and PA graphs (Figs. \ref{exp:number-AND}-\ref{exp:number-NOT}). This increase is caused by a larger search space and resulting hash collisions. Moreover, queries in ER graphs take more time than in PA graphs. This difference arises because ER graphs have a consistent degree distribution, whereas PA graphs have a degree distribution that is uneven. ER graphs with a uniform distribution of vertex out-degrees possess more paths between vertex pairs than PA graphs with a skewed out-degree distribution. This surplus of paths leads to additional checks when answering PCR queries, consequently prolonging query times. 

\section{Related Work}\label{appendix:relatedwork}
There has been a lot of research efforts on reachability queries \cite{2023'overview}. Early research work focuses mainly on unlabeled graph. Since many real-world graphs are labeled on edges, recent efforts try to answer reachability queries with label constraints. 

\subsection{Unlabeled Graph}
Answering a reachability query on unlabeled directed graphs is to find a path from the source vertex to the target vertex. If there exists a path, then the answer is reachable (true) and otherwise unreachable (false). 
Many approaches \cite{2006'Dual, 2009'3-hop, 2011'Path-tree, chen'2008'chain, 2005'HLSS, 2005'HOPI, 2010'GRAIL, 2012'GRAIL, 2007'GRIPP, 2013'ferrari, 2014'IP, 2018'IP, 2013'TFLabel, 1989'TCC, 2017'BFL} have been proposed to answer reachability queries on unlabeled graphs. 
H. Wei et al. \cite{2018'IP} divided the methods into Label-Only and Label+G where label on each vertex indicates full or partial reachability information among vertices. If labels only have partial reachability information, Label+G algorithms have to traverse the graph in order to answer the queries, such as \textit{GRAIL} \cite{2010'GRAIL,2012'GRAIL} and \textit{Ferrari} \cite{2013'ferrari}, while the Label-Only algorithms like \textit{Dual-Label} \cite{2006'Dual}, \textit{3-Hop} \cite{2009'3-hop} and \textit{TF-Label} \cite{2013'TFLabel} are able to answer directly through the labels. 
Yuan et al. \cite{2018'MGTag,2022'MGTag} classify the approaches into dimension labeling and set labeling. Labels assigned by dimension labeling methods \cite{chen'2008'chain,2006'Dual,2008'pathtree,2011'Path-tree,2011'PWAH,2010'GRAIL,2012'GRAIL,2014'feline,2013'ferrari} can show relative topological relationships among vertices in different dimensions. For example, Y. Chen's algorithm \cite{chen'2008'chain} and \textit{Path-tree} \cite{2011'Path-tree,2008'pathtree} are from the dimension of chain decomposition, and MGTag \cite{2018'MGTag,2022'MGTag} is according to subgraphs and layers. So, we can answer some queries by checking the topological relationships indicated in the dimension labels of two vertices. 
The set labeling methods \cite{2005'HOPI,2006'SCI,2005'HLSS,2014'IP,2018'IP,2009'3-hop,2010'Path-hop,2013'TFLabel,2013'HL,2017'BFL} are based on 2-hop labels. For example, BFL \cite{2017'BFL} maps all the vertices that are reachable from vertex $u$ into a bit set (Out($u$)), and maps all the vertices that can reach $u$ into another bit set (In($u$)). When answering queries, BFL checks whether target $v$ is not in Out($u$) of $u$ or $u$ is not in In ($v$). If the answer is yes, then $u$ and $v$ are unreachable. Otherwise, BFL will traverse the graph and repeat the above check for each vertex until the answer can be given. Since a large graph is sparse, the algorithm can quickly prune all unreachable branches. 

\subsection{Labeled Graph}
A reachable query on a labeled graph requires not only the existence of a path between two given vertices, but also the sequence/set of labels on the path to match given constraints, which are commonly specified by regular expressions. Regular expression patterns can be classified into three types \cite{2019'ARRIVAL}: label-set restricted paths, repeated label-sequence paths, and concatenated label-chains. 
Label-set restricted paths are typically known as Label-Constrained Reachability (LCR) queries.
\subsubsection{Regular Path Query}
A Regular Path Query (RPQ) specifies that the labels of any solution paths must satisfy a given regular expression and the problem is known to be NP-Hard \cite{1995'RPQNPHard}. There are several algorithms about RPQ. 
A. Koschmieder et al. devise an algorithm to answer RPQs by decomposing an RPQ into several smaller RPQs using infrequent labels \cite{2012'RPQonLG}. However, the algorithm depends on rare labels and can not work in all cases (e.g., labels with similar frequency). 
ARRIVAL \cite{2019'ARRIVAL} samples a number of paths through bi-directional random walk. If there exists any sampling path between two vertices, the two vertices are reachable, otherwise unreachable. So, the answer may be false-negative. 
Similarly, the example-based regular path query (RQuBE) \cite{2021'RQuBE} also employs a sampling-based method to build a candidate vertex set and return top-k vertices based on their confidence values as the final result set. Furthermore, RQuBE can automatically infer regular expressions from user-provided exemplars, making it user-friendly for individuals with limited knowledge on regular expression. Unlike previous sampling algorithms, D. Arroyuelo et al. proposed a new algorithm in \cite{2022'Ring} that combines bit-parallel simulation and the ring index to process automaton states synchronously. In addition to generating vertices pairs from the constraints of regular expressions, there is some other work on regular path queries. For example, A. Pacaci et al. \cite{2020'Streaming} determine whether a given constraint is satisfied between two concrete entities over streaming graphs. Recent RLC queries \cite{2023'RLC} require solution paths consisting of one or more concatenations of the given sequence. 

\subsubsection{Label-constrained Reachability Query} LCR specifies a set of labels and restricts the labels on the solution paths to only those within this set. 
As one of the most common queries on reachability in labeled graphs, several
algorithms have also been proposed to address this problem. 
LI+ \cite{2017'LI} selects a small number of landmark vertices and precomputes all pairs of vertex that can be reached via the landmarks. Then, LI+ can answer LCR queries using the stored information. Y. Chen et al. \cite{Chen'2021'Recurve} propose an algorithm that recursively decomposes the input graph while transforming the query into a series of subqueries to answer LCR queries. 
The state-of-the-art algorithm P2H+ \cite{2020'P2H,2022'P2H} introduces a complete reachable index based on the 2-hop cover framework to answer LCR. Specifically, for each vertex $u$, $Out(u)$ stores all vertices that $u$ can reach, along with the minimal set of labels on the path between $u$ and each reachable vertex. Similarly, $In(u)$ stores the corresponding information for the vertices that can reach $u$.
Y. Cai et al. \cite{2023'PDU} improve P2H+ by reducing the index size for vertices with one degree and eliminating unreachable queries without label constraints, denoted as PDU (P2H+DOR+UQF). 
Since existing algorithms for LCR build complete indices, they offer the best performance on queries. However, their indexing costs (e.g. index size and indexing time) are relatively high, making them impractical for construction on large graphs. 


LCR queries restrict the reachability paths to only the edges that have labels in the given set of labels. RPQ imposes strict requirements on the label sequence based on the constraints of regular expressions. However, in real-world scenarios, users may expect more combinations of label sets (or patterns). Therefore, we propose composite patterns using logical operators that relaxes these label constraints for more flexibility. 
In contrast to previous approaches, we construct a particle index answering whether a given vertex pair is reachable under a specified pattern.

\balance

\end{document}